\documentclass[runningheads,a4paper]{llncs}
\usepackage{amsmath}
\usepackage{amssymb}
    \usepackage{amsfonts}
    \usepackage{amsxtra}
    \usepackage{amstext}
    \usepackage{latexsym}
    \usepackage{dsfont} 

\usepackage{graphicx}

\usepackage{url}
\newcommand{\keywords}[1]{\par\addvspace\baselineskip
\noindent\keywordname\enspace\ignorespaces#1}

\begin{document}
\mainmatter  

\title{Teleportation-based quantum homomorphic encryption scheme with quasi-compactness and perfect security}

\titlerunning{Teleportation-based quantum homomorphic encryption scheme}

%

\author{Min Liang
\thanks{This work was supported by the National Natural Science Foundation of China (Grant No. 61672517).}
}

\authorrunning{Min Liang}

\institute{Data Communication Science and Technology Research Institute, Beijing 100191, China\\
              \email{liangmin07@mails.ucas.ac.cn}}

%
%

\maketitle

\vspace{-5mm}
\begin{abstract}
Quantum homomorphic encryption (QHE) is an important cryptographic technology for delegated quantum computation. It enables remote Server performing quantum computation on encrypted quantum data, and the specific algorithm performed by Server is unnecessarily known by Client. Quantum fully homomorphic encryption (QFHE) is a QHE that satisfies both compactness and $\mathcal{F}$-homomorphism, which is homomorphic for any quantum circuits. However, Yu et al.[Phys. Rev. A 90, 050303(2014)] proved a negative result: assume interaction is not allowed, it is impossible to construct perfectly secure QFHE scheme.
So this article focuses on non-interactive and perfectly secure QHE scheme with loosen requirement, specifically quasi-compactness.

This article defines encrypted gate, which is denoted by $EG[U]:|\alpha\rangle\rightarrow\left((a,b),Enc_{a,b}(U|\alpha\rangle)\right)$. We present a gate-teleportation-based two-party computation scheme for $EG[U]$, where one party gives arbitrary quantum state $|\alpha\rangle$ as input and obtains the encrypted $U$-computing result $Enc_{a,b}(U|\alpha\rangle)$, and the other party obtains the random bits $a,b$. Based on $EG[P^x](x\in\{0,1\})$, we propose a method to remove the $P$-error generated in the homomorphic evaluation of $T/T^\dagger$-gate. Using this method, we design two non-interactive and perfectly secure QHE schemes named \texttt{GT} and \texttt{VGT}. Both of them are $\mathcal{F}$-homomorphic and quasi-compact (the decryption complexity depends on the $T/T^\dagger$-gate complexity).

Assume $\mathcal{F}$-homomorphism, non-interaction and perfect security are necessary property, the quasi-compactness is proved to be bounded by $O(M)$, where $M$ is the total number of $T/T^\dagger$-gates in the evaluated circuit. \texttt{VGT} is proved to be optimal and has $M$-quasi-compactness.
According to our QHE schemes, the decryption would be inefficient if the evaluated circuit contains exponential number of $T/T^\dagger$-gates. Thus our schemes are suitable for homomorphic evaluation of any quantum circuit with low $T/T^\dagger$-gate complexity, such as any polynomial-size quantum circuit or any quantum circuit with polynomial number of $T/T^\dagger$-gates.

\keywords{quantum cryptography, quantum encryption, delegated quantum computation, quantum homomorphic encryption}
\end{abstract}

\vspace{-8mm}
\section{Introduction}
Modern cryptography can ensure the security of communication. However, in the era of cloud computing, the computation is also required to be protected using cryptographic technology. As a key technology to guarantee the security of classical computation, fully homomorphic encryption (FHE) has been studied extensively since Gentry's breakthrough result \cite{Gentry2009}. With the development of quantum computer, quantum computation would be practical in the near future, and provide quantum computing service for clients through the cloud.
In order to ensure the security of cloud quantum computing, the researchers propose mainly two classes of quantum schemes:
(1) blind quantum computation (BQC), and (2) quantum homomorphic encryption (QHE). Both of them are two-party delegated quantum computing technology.

BQC enables a client with weak quantum ability to delegate a quantum task to remote server with strong quantum ability, where the client's input and final result are kept secret and the quantum program is only known by the client.
QHE allows remote server to perform some quantum program $QC$ on encrypted quantum data $Enc(\rho)$ provided by a client, and the client can decrypt server's output and obtain the result $QC(\rho)$.
By comparing the functionality of BQC and QHE, the fundamental difference is that: BQC finishes a quantum program decided only by Client and Server cannot obtain Client's quantum program; however, in QHE, the quantum program is decided by Server and it is unnecessary for Client to know the quantum program (actually it may be negotiated by both parties in advance). So BQC is more powerful than QHE and it is harder to design BQC scheme. When interaction is allowed, the design of BQC protocol has been tackled very well, and lots of results have been proposed in these literatures (\cite{Childs2005,Aharonov2010,Broadbent2009,Sueki2013,Morimae2012a,Morimae2012b,Giovannetti2013}). Whether or not a non-interactive scheme exists for BQC? No result has been given at present.

In modern cryptography, FHE is usually achievable with non-interactive scheme. Similarly, quantum fully homomorphic encryption (QFHE) is also expected to be a non-interactive quantum algorithm. Actually, interaction can decrease the difficulty of the construction of QFHE scheme. For example, the scheme in Ref.\cite{Liang2015a} uses a few interactions, and its total number is the same as that of $T$-gates. Because the interactions necessarily decrease the efficiency of the protocol, this article focuses on interactive QHE scheme.

The early exploration of QHE begins in 2012. Rohde et al.\cite{Rohde2012} propose a simple symmetric-key QHE scheme, which allows quantum random walk on encrypted quantum data. Their scheme is not $\mathcal{F}$-homomorphic, that means it does not allow homomorphic evaluation of arbitrary quantum circuits. Later, Ref.\cite{Liang2013} defines QFHE and constructs a weak scheme, which allows local quantum computation on encrypted data without access any plain data. This scheme does not accord with the usual concept of FHE, and its applicable situation is very limited.
Ref.\cite{Tan2016} proposes a QHE scheme for a large class of unitaries, but it cannot provide the security of cryptographic sense: it can hide only $n/\log n$ bits ($n$ is the size of the message), so the ratio of the hidden amount to the total amount approaches to 0 ($1/\log n\rightarrow 0$) when $n$ approaches to infinity.

Yu et al.\cite{Yu2014} prove a no-go result: if non-interaction is required, any QFHE scheme with perfect security must use exponential size of storage; so it is impossible to construct an efficient QFHE scheme with perfect security. Later, the researchers study QHE schemes with certain loosen condition.
For example, the interactivity is allowed in the schemes \cite{Liang2015a,Fisher2014}, the expected security level is computational security \cite{Broadbent2015,Dulek2016}, the client accommodates the server with enough copies of encrypted ancillary states \cite{Liang2015b,Ouyang2015}.
The scheme in Ref.\cite{Liang2015a} is perfectly secure, however, it needs a few interactions, whose number is the same as
the number of $T$-gates in universal quantum circuit. Though the schemes in Ref.\cite{Broadbent2015} are non-interactive, they are constructed by the combination of quantum one-time pad (QOTP) and classical (fully) homomorphic encryption; so their efficiency and security would be bounded by classical homomorphic encryption scheme.
The scheme proposed by Dulek et al.\cite{Dulek2016} is also non-interactive. In their scheme, it is necessary to use a ancillary gadget, in which the classical FHE scheme is used; so the security of the whole scheme would also be limited to the computational security of classical FHE scheme.
Refs.\cite{Liang2015b,Ouyang2015} associate homomorphic encryption with transversal computation, and propose non-interactive QFHE schemes based on quantum codes. In their schemes, Client must provide Server with sufficient copies of certain encrypted ancillary state; Too many copies of the ancillary states would leak some information about the secret key, so this kind of construction cannot achieve perfect security.

Besides the no-go result of Yu et al.\cite{Yu2014}, Newman and Shi \cite{Newman2018} and Lai and Chung \cite{Lai2018} independently prove an enhanced no-go result: it is impossible to construct non-interactive and information theoretically secure (ITS) QFHE scheme. According to their results, a QFHE scheme can achieve computational security at best if non-interaction is required. Then it is very natural to ask that: whether QHE scheme can provide quantum advantage than its classical counterpart? Lai and Chung \cite{Lai2018} construct a non-interactive ITS QHE scheme, which has compactness. However, it is not $\mathcal{F}$-homomorphic but only partially homomorphic\footnote{A QHE scheme which is compact and not $\mathcal{F}$-homomorphic is called quantum somewhat homomorphic encryption scheme. This kind of QHE scheme is not studied in this article. In fact, it is also worth to study somewhat QHE scheme which is homomorphic for a sufficiently large class of quantum circuits.}. Concretely, it only allows the homomorphic evaluation of a class of quantum circuits called instantaneous quantum polynomial-time (IQP) circuits¡£

We consider non-interactive and perfectly secure QHE scheme, which is $\mathcal{F}$-homomorphic but not compact. Here QHE is not compact if its decryption procedure has dependence on the evaluated circuit. It is still worth to study QHE without compactness if decryption procedure has complexity that scales sublinearly in the size of evaluated circuit. This kind of property is called quasi-compactness \cite{Broadbent2015}. This article will focus on QHE schemes with both $\mathcal{F}$-homomorphism and quasi-compactness.

\subsection{Our Contributions}
Generally, a QFHE scheme should be compact and can evaluate any quantum circuits homomorphically. Broadbent and Jeffery \cite{Broadbent2015} define quasi-compactness for QHE following the quasi-compactness of FHE in Ref.\cite{Gentry2009}. They design a QHE scheme named \texttt{EPR} with quasi-compactness, which is also non-interactive and computationally secure.
Then, if quasi-compactness is expected, whether there exists non-interactive and perfectly secure QHE scheme? This article will give a positive answer.
Concretely, our contributions are listed as follows.
\begin{enumerate}
  \item We define a new fundamental module named encrypted gate $EG[U]:|\alpha\rangle\rightarrow\left((a,b),Enc_{a,b}(U|\alpha\rangle)\right)$, where $a,b\in\{0,1\}$ are random bits, and $Enc_{a,b}$ represents quantum encryption transformation with the keys $a,b$. Based on gate teleportation, $EG[U]$ can be implemented by a two-party computation scheme, where one party obtains the classical part $a,b$ and the other party obtains the encrypted $U$-computing result (e.g.$Enc_{a,b}(U|\alpha\rangle)$).
  \item In order to construct a QHE scheme, the main difficulty is that the homomorphic evaluation of $T/T^\dagger$-gate would cause a $P$-error ($P^\dagger$ in the formula $TX^xZ^z|\phi\rangle=(P^\dagger)^xX^xZ^{x\oplus z}T|\phi\rangle$) which is hard to remove. Ref.\cite{Liang2015a} removes $P$-error using one round of interaction, and Ref.\cite{Dulek2016} removes $P$-error using an ancillary gadget with computational security. This article proposes a new method to remove $P$-error based on encrypted gate $EG[P^x]$: once a $T/T^\dagger$-gate is executed, Server together with Client should finish an encrypted gate $EG[P^x]$, where Client only performs quantum measurement; Based on principle of deferred measurement, all the measurements performed by Client can be deferred until Server finishes the homomorphic evaluation of the circuit, thus the interactions are avoided in the evaluation procedure.
  \item By making use of our method to remove $P$-error, a QHE scheme \texttt{GT} and a variation \texttt{VGT} are constructed. Both the two QHE schemes have the four properties: non-interactive, perfect security, $\mathcal{F}$-homomorphism (or homomorphic evaluation of any quantum circuits), and quasi-compactness. The complexity of decryption depends on the total number of $T/T^\dagger$-gates in the evaluated circuit.
  \item It is proved that, for any non-interactive and perfectly secure QHE scheme with $\mathcal{F}$-homomorphism, the optimal compactness is bounded by $M$-quasi-compactness. The first scheme \texttt{GT} can achieve $M\log M$-quasi-compactness, and the slightly modified version \texttt{VGT} can achieve the optimal bound.
\end{enumerate}
It is worth to stress that: though our schemes are quasi-compact, it is still very efficient in homomorphic evaluation of any quantum circuits with low $T/T^\dagger$-gate complexity. If there are too many $T/T^\dagger$-gates in the evaluated circuit, then the decryption procedure would be inefficient. Here $M\log M$ or $M$-quasi-compactness can asymptotically describe how fast the decryption complexity increases along with $T/T^\dagger$-gates complexity. We hope the value of quasi-compactness is as small as possible. Then the QHE scheme can implement larger circuits if Client's computational capability is fixed.

\subsection{Related works}
The QHE scheme \texttt{EPR} proposed by Broadbent and Jeffery \cite{Broadbent2015} makes use of Bell state and quantum measurement. That scheme is constructed by the combination of QOTP and classical FHE, then it is computational secure. Moreover, it is proved to be $M^2$-quasi-compact, where $M$ is the number of $T$-gates in an evaluated circuit. In this article, our schemes also make use of Bell state and quantum measurement. However, our scheme \texttt{VGT} has perfect security and $M$-quasi-compactness, so it is better than \texttt{EPR}.

Yu et al.\cite{Yu2014} prove a no-go result: if interaction is not allowed, there does not exist QFHE scheme with perfect security. An enhanced no-go result has been proved independently by Newman and Shi \cite{Newman2018} and Lai and Chung \cite{Lai2018}: if interaction is not allowed, there does not exist ITS QFHE scheme.
This article focuses on the non-interactive and perfectly secure QHE schemes. Our schemes are quasi-compact, so they are not QFHE schemes, and then our result does not contradict with those no-go results.  Though our schemes are not QFHE schemes, they can implement any unitary quantum circuit homomorphically (maybe inefficiently).

Alagic et al.\cite{Alagic2017} study computationally secure QFHE scheme, and define an additional property named verification. This article does not consider QHE with verification.

\subsection{Organization}
Section 2 introduces some preliminary knowledge about quantum computation, including gate teleportation. In Section 3, encrypted gate is defined and a method to remove $P$-error is proposed, then the homomorphic scheme for $T/T^\dagger$-gate is constructed. Section 4 contains the detailed QHE schemes and analysis.

\section{Preliminaries}
\subsection{Quantum computation}
For a detailed introduction about quantum computation, we refer the reader to Nielsen and Chuang \cite{Nielsen2000}. Here we give a brief overview of some elementary quantum logic gates and notations.

Pauli $X$ and $Z$ gates are described as unitary operators
$X=\left(
     \begin{array}{cc}
       0 & 1 \\
       1 & 0 \\
     \end{array}
   \right)
$, and
$Z=\left(
     \begin{array}{cc}
       1 & 0 \\
       0 & -1 \\
     \end{array}
   \right)
$.
Hadamard gate is $H=\frac{1}{\sqrt{2}}(X+Z)$, and the phase gate is $P=\sqrt{Z}$. Controlled-NOT gate is
$CNOT=\left(
        \begin{array}{cccc}
          1 & 0 & 0 & 0 \\
          0 & 1 & 0 & 0 \\
          0 & 0 & 0 & 1 \\
          0 & 0 & 1 & 0 \\
        \end{array}
      \right)
$, and implements the following quantum transformation $$CNOT|c\rangle|t\rangle=|c\rangle|t\oplus c\rangle,\forall c,t\in\{0,1\}.$$
$T$-gate is
$T=\left(
     \begin{array}{cc}
       1 & 0 \\
       0 & e^{i\frac{\pi}{4}}\\
     \end{array}
   \right)
$. Among these quantum gates, $T$-gate is the only one non-Clifford gate. In the construction of QHE scheme, the main difficulty is how to devise the homomorphic evaluation of $T$-gate. The conjugate of $T$ is
$T^\dagger=\left(
             \begin{array}{cc}
               1 & 0 \\
               0 & e^{-i\frac{\pi}{4}} \\
             \end{array}
           \right)
$. All these quantum gates form a set $\mathcal{S}=\{X,Z,H,P,CNOT,T,T^\dagger\}$.

Because $T^\dagger=T^7$, then $T^\dagger$ can be implemented by seven $T$-gates. In order to simplify the description of quantum circuit, $T$ and $T^\dagger$ are used together in many practical quantum circuits. In the next section, we construct QHE scheme for $T$-gate and introduce how to modify it so as to implement $T^\dagger$ homomorphically.

For any unitary transformation $U$ and any $n$-bit string $b=b_1b_2\cdots b_n$, denote $U^b=\bigotimes_{i=1}^nU^{b_i}$, especially $U^{11\cdots 1}=\bigotimes_{i=1}^nU=U^{\otimes n}$.

For two $n$-bit strings $a,b\in\{0,1\}^n$, we define $a\oplus b=(a_1\oplus b_1,a_2\oplus b_2,\ldots,a_n\oplus b_n)$.

From Pauli operators $X$ and $Z$, a Pauli group can be generated such as $\{X^aZ^b: a,b\in\{0,1\}^n\}$. Since
\begin{equation}\label{eq1}
\frac{1}{2^{2n}}\sum_{a,b\in\{0,1\}^n}X^aZ^b\sigma(X^aZ^b)^\dagger=\frac{I_{2^n}}{2^n},
\end{equation}
where $\sigma$ is arbitrary quantum state and $\frac{I_{2^n}}{2^n}$ is a $n$-qubit totally mixed state. According to Eq.(\ref{eq1}), Boykin and Roychowdhury \cite{Boykin2003} and Ambainis et al.\cite{Ambainis2000} propose QOTP, which is a symmetric-key encryption scheme for quantum states. It is a kind of quantum cryptographic primitive with perfect security.

Quantum state space is usually described as a complex Hilbert space. In the quantum-message-oriented encryption scheme, the plaintext space and ciphertext space are complex Hilbert space, denoted as $\mathcal{H}_M$ and $\mathcal{H}_C$ respectively. We denote the set of density operators on the space $\mathcal{H}_M$ (or $\mathcal{H}_C$) by $D(\mathcal{H}_M)$ (or $D(\mathcal{H}_C)$). For a quantum state that is pure state, it can be described as a unit vector in complex Hilbert space $\mathcal{H}$; it can also be described as a density operator on space $\mathcal{H}$. In the remainder of this article, a quantum state may be a unit vector or a density operator, and we will not state it definitely since the readers can tell them apart very easily.

Given two quantum states $\rho$ and $\sigma$,  their trace distance is defined as $\Delta(\rho,\sigma)=Tr(|\rho-\sigma|)$, where $|A|$ is defined as $\sqrt{A^\dagger A}$. A quantum channel $\Phi:D(\mathcal{A})\rightarrow D(\mathcal{B})$ is a physically-realizable mapping on quantum registers. Given a quantum circuit $QC$ with $n$-qubit input and $m$-qubit output, a quantum channel $\Phi_{QC}$ can be induced as a quantum transformation from $n$ qubits to $m$ qubits.

\subsection{Gate teleportation}
Gate teleportation can be seen as an extension of quantum teleportation. It is mainly applied in measurement-based quantum computing \cite{Gottesman1999,Jozsa2005}.

Define EPR entanglement state $|\Phi_{00}\rangle=\frac{1}{\sqrt{2}}(|00\rangle+|11\rangle)$, then the four Bell states can be described as
\begin{equation}
|\Phi_{ab}\rangle=(Z^bX^a\otimes I)|\Phi_{00}\rangle,\forall a,b\in\{0,1\}.
\end{equation}

For any single-qubit gate $U$, define ``$U$-rotated Bell basis" as the following set
\begin{equation}
\Phi(U)=\{|\Phi(U)_{ab}\rangle,a,b\in\{0,1\}\},
\end{equation}
where quantum state $|\Phi(U)_{ab}\rangle=(U^\dagger\otimes I)|\Phi_{ab}\rangle$ or $|\Phi(U)_{ab}\rangle=(U^\dagger Z^bX^a\otimes I)|\Phi_{00}\rangle$.

For a single-qubit state $|\alpha\rangle$, the following formula holds \cite{Jozsa2005}:
\begin{equation}\label{eq2}
|\alpha\rangle\otimes|\Phi_{00}\rangle=\sum_{a,b\in\{0,1\}}|\Phi_{ab}\rangle\otimes X^aZ^b|\alpha\rangle.
\end{equation}

In fact, Eq.(\ref{eq2}) is just a mathematical description of quantum teleportation. It can be extended as follows:
\begin{equation}\label{eq3}
|\alpha\rangle\otimes|\Phi_{00}\rangle=\sum_{a,b\in\{0,1\}}|\Phi(U)_{ab}\rangle\otimes X^aZ^bU|\alpha\rangle,
\end{equation}
where $U$ is any single-qubit gate. According to Eq.(\ref{eq3}), if one carries out ``$U$-rotated Bell measurement" (that is, $U$-rotated Bell basis $\Phi(U)$ is selected as the measurement basis in the quantum measurement) on the qubit $|\alpha\rangle$ and the first qubit of $|\Phi_{00}\rangle$, and obtains the result $a,b$, then the second qubit of $|\Phi_{00}\rangle$ would collapse into quantum state $X^aZ^bU|\alpha\rangle$. Finally, performing Pauli operator on $X^aZ^bU|\alpha\rangle$ would result in the state $U|\alpha\rangle$.

Similar to quantum teleportation, we can describe the gate teleportation: Alice and Bob preshare a pair of qubits in the state $|\Phi_{00}\rangle$; Alice prepares a qubit in the state $|\alpha\rangle$, and performs $U$-rotated Bell measurement on Alice's local two qubits, then announces the results; According to the Alice's results, Bob performs Pauli opertors ($X$ or $Z$ or both) and obtains a new state $U|\alpha\rangle$. The gate teleportation can be described as Fig.\ref{fig1}.

\begin{figure}[hbtp]
\centering
  \includegraphics[scale=1.1]{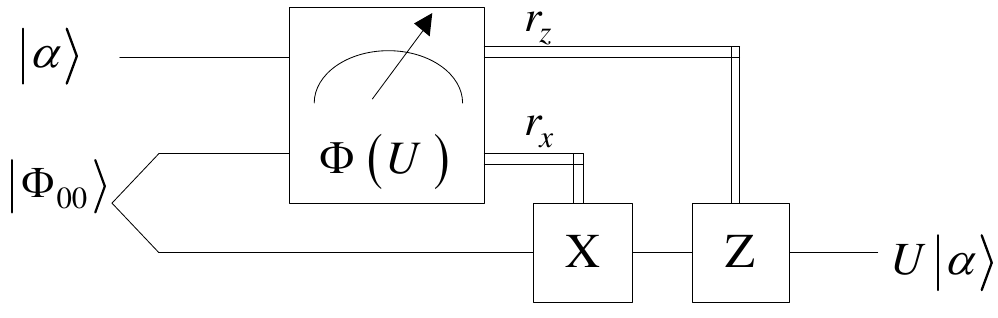}
\caption{Quantum circuit for gate teleportation. In the quantum measurement, $U$-rotated Bell basis $\Phi(U)$ is the measurement basis.}
\label{fig1}       
\end{figure}

Obviously, gate teleportation is an extension of quantum teleportation. If $U$ is the identity, then $\Phi(U)$ would be the standard Bell basis, and gate teleportation would degenerate into quantum teleportation.

\section{Homomorphic evaluation of $T$-gate}
In the set of universal quantum gates, $T$-gate is the only one non-Clifford gate. In order to construct QHE scheme, the main obstacle is how to homomorphically evaluate the $T$-gate. The QHE scheme in Ref.\cite{Liang2015a} realizes $T$-gate homomorphically, however an interaction are used to remove $P$-error that occurs in the homomorphic evaluation of $T$-gate: whenever Server performs a $T$-gate on a qubit, he sends the qubit to Client; Then Client performs a quantum operator $P^x$ ($x$ is a secret bit which is known only by Client)\footnote{In the concrete scheme, in order to protect $x$ from the eavesdropping of Server, Client should perform quantum operator $X^rZ^{r'}P^x$, where both $r$ and $r'$ are secret random bits chosen by Client.} on the qubit and returns it to Server.

In this section, a new module named $EG[P^x]$ is defined, and a new method based on $EG[P^x]$ is proposed to remove $P$-error. Based on gate teleportation, the module $EG[P^x]$ is implemented through a two-party secure computation scheme, which contains quantum measurement. When the module $EG[P^x]$ is used in the homomorphic evaluation of $T$-gate, Client can defer all the quantum measurements until Server has finished all his/her computation. Thus the interactions can be avoided.

\subsection{Encrypted gate based on gate teleportation}
Define encrypted gate $EG[U]:|\alpha\rangle\rightarrow\left((a,b),Enc_{a,b}(U|\alpha\rangle)\right)$, where $a,b\in\{0,1\}$ are random bits,
$Enc_{a,b}$ is a quantum encryption operator with the keys $a,b$. Especially, if $Enc_{a,b}$ is a QOTP encryption, then $EG[U]:|\alpha\rangle\rightarrow\left((a,b),X^aZ^bU|\alpha\rangle\right)$.
In order to apply it to the construction of QHE, we present a two-party computation scheme for $EG[U]$ based on gate teleportation. In our scheme, Server provide a quantum state $|\alpha\rangle$ as input; While the computation of $EG[U]$ is finished, Client will obtain random bits $(a,b)$ and Server obtains the quantum state $X^aZ^bU|\alpha\rangle$; Moreover, $U$ is selected by Client, and Server cannot obtain any information about $a,b$ and $U$. For any two-party computation scheme implementing $EG[U]$, if Server can intercept zero information about $a,b$ and $U$, then we say the scheme has perfect security.

The most important feature of encrypted gate $EG[U]$ is that the computing of $U$-gate is encrypted and the gate $U$ must be secure (e.g. in the two-party computation scheme, Server cannot know which gate is selected by Client).

Fig.\ref{fig2} describes a gate-teleportation-based implementation of $EG[U]:|\alpha\rangle\rightarrow\left((r_x,r_z),X^{r_x}Z^{r_z}U|\alpha\rangle\right)$,
where $r_x,r_z\in\{0,1\}$ are random bits. This implementation must use the Bell state and quantum measurement.

\begin{figure}[hbt]
\centering
  \includegraphics[scale=1.6]{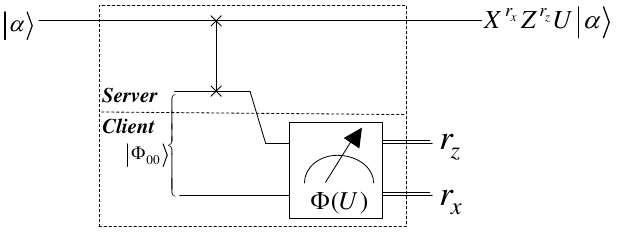}
\caption{Encrypted gate $EG(U)$ is implemented by a two-party computation scheme based on gate teleportation. The dashed frame at the top shows Server's operations, and the dashed frame at the bottom shows Client's operations. There is only one qubit transmission from Server to Client.}
\label{fig2}       
\end{figure}

The detailed procedure is as follows.
\begin{enumerate}
  \item Server and Client preshare quantum entanglement $|\Phi_{00}\rangle$;
  \item Server sends quantum state $|\alpha\rangle$ to Client;
  \item Client performs $U$-rotated Bell measurement, and obtains random results $r_x,r_z\in\{0,1\}$; Then the qubit retained by Server would be in the state $X^{r_x}Z^{r_z}U|\alpha\rangle$.
\end{enumerate}

Especially when $U=P^x$,
\begin{equation}
|\alpha\rangle\otimes|\Phi_{00}\rangle=\sum_{a,b\in\{0,1\}}|\Phi(P^x)_{ab}\rangle\otimes X^aZ^bP^x|\alpha\rangle,\forall x.
\end{equation}
Considering $P^x$-rotated Bell measurement (or quantum measurement in the basis $\{|\Phi(P^x)_{ab}\rangle\}_{a,b\in\{0,1\}}$) on the first two qubits, if one measures two bits $a$ and $b$, then the third qubit would be in the state $X^aZ^bP^x|\alpha\rangle$.

According to gate teleportation, it can be easily shown the correctness of the implementation of $EG[U]$ in Fig.\ref{fig2}. We focus on its security in the following.

\begin{proposition}\label{prp1}
The two-party computation scheme of $EG[U]$ in Fig.\ref{fig2} is perfectly secure. That is, Server cannot obtain any information about $U$ and $r_x,r_z$.
\end{proposition}
\begin{proof}
In fact, the dashed frame shows a slightly modified version of gate teleportation, so Client would obtain uniformly random bits $r_x,r_z$ and Server cannot know any information about them. If the dashed frame is viewed as a black box (both Client and Server contribute to the functionality of the black box), then one can see that Client performs quantum operation $X^{r_x}Z^{r_z}U$ on Server's quantum state $|\alpha\rangle$ and the two bits $r_x,r_z\in\{0,1\}$ are unknown by Server. Thus Server does not know which gate is selected by Client, and its perfect security is derived from QOTP. Especially when $U=P^x$, the scheme in Fig.\ref{fig2} realizes the encrypted gate $EG[P^x]$. At the beginning of its execution, Server does not know which basis (or the value of $x$) is chosen by Client. At the end, Server cannot deduce Client's measurement basis (or the value of $x$) from all the inputs and outputs; The reason is as follow: Client's measurement results $r_x,r_z$ are completely random and are unknown by Server, so $X^{r_x}Z^{r_z}$ can be viewed as a QOTP encryption, then the value of $x$ is perfectly hidden in the output $X^{r_x}Z^{r_z}P^x|\alpha\rangle$.~~\hfill{}$\Box$
\end{proof}

Next, we propose the detailed construction of QHE scheme for $T$-gate. In the scheme, a method to remove $P$-error is introduced based on the encrypted gate $EG[P^x]$, then no interaction is necessary in the homomorphic evaluation of $T$-gate.

\subsection{QHE scheme for $T$-gate}
QHE scheme for $T$-gate contains the following five parts: \textbf{Setup}, \textbf{Key Generation}, \textbf{Encryption}, \textbf{Evaluation of $T$-gate}, and \textbf{Decryption}, where \textbf{Setup} preshares Bell state between Client and Server, and \textbf{Key Generation} generates encryption key (or secret key). The decryption key is generated until the key-updating procedure in \textbf{Decryption}. In the \textbf{Decryption}, Client should perform quantum measurement using the secret key; Then Client updates the secret key according to the key-updating rule and measurement result; Finally, Client obtains the decryption key (or final key).

In the following scheme, the first qubit of $|\Phi_{00}\rangle$ is relabeled as $c$, and the other as $s$. Then this Bell state is denoted by $|\Phi_{00}\rangle_{c,s}$.

\noindent\fbox{
\begin{minipage}{33.3em}
\begin{center}
    \textbf{QHE scheme for $T$-gate}
\end{center}
\begin{description}
  \item[1.Setup:] EPR source generates Bell state $|\Phi_{00}\rangle_{c,s}$, where qubit $c$ and qubit $s$ are held by Client and Server, respectively.
  \item[2.Key Generation:] Generate random bits $x,z\in\{0,1\}$, and output $sk=(x,z)$ as secret key.
  \item[3.Encryption:] Given a single-qubit input, Client performs QOTP encryption with secret key $sk=(x,z)$. For example $|\alpha\rangle\rightarrow X^xZ^z|\alpha\rangle$.
  \item[4.Evaluation of $T$-gate:] Server performs $T$-gate on encrypted qubit, and then a quantum $\mathrm{SWAP}$ operation on that qubit and qubit $s$, and finally outputs both of them.
  \item[5.Decryption:] According to secret key $sk=(x,z)$, Client carries out quantum measurement and key-updating, and obtains a new key which can be used as the decryption key.

  \begin{minipage}{31.5em}
  \begin{description}
    \item[5-1.Quantum measurement.] According to the key $x$, Client performs $P^x$-rotated Bell measurement on qubit $s$ and qubit $c$, and obtains random values $(r_x,r_z)$.
    \item[5-2.Key-updating.] From the key-updating rules (\textbf{Appendix C}) and the secret key $sk=(x,z)$ and measurement result $(r_x,r_z)$, Client computes the new key $(x',z')=(x\oplus r_x,x\oplus z\oplus r_z)$.
    \item[5-3.QOTP decryption.] With the new key $(x',z')$, Client performs QOTP decryption on the encrypted result and obtains the plaintext result. For example $|\varphi\rangle\rightarrow Z^{z'}X^{x'}|\varphi\rangle=T|\alpha\rangle$.
  \end{description}
  \end{minipage}
\end{description}
\end{minipage}
}

QHE scheme for $T$-gate can be described as Fig.\ref{fig3}.

\begin{figure}[hbt]
\centering
  \includegraphics[scale=1.2]{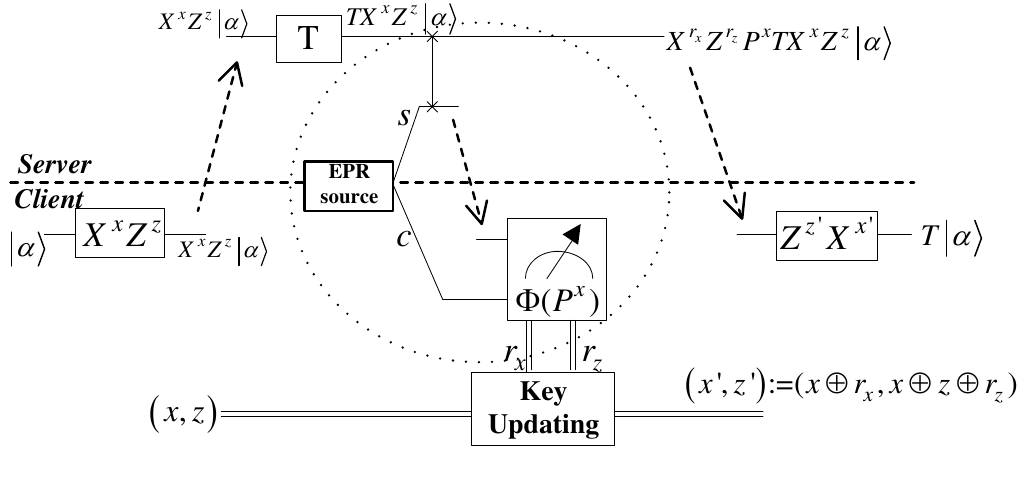}
\caption{QHE scheme for $T$-gate. Server's operations are above the dashed line and Client's operations are below the dashed line. Each arrow shows the direction in which the information is transmitting, e.g. from Server to Client or from Client to Server.
The dashed circle describes the two-party computation scheme of $EG[P^x]$. Key-updating is a classical algorithm.}
\label{fig3}       
\end{figure}

In Fig.\ref{fig3}, the dashed circle shows the two-party computation scheme of encrypted gate $EG[P^x]$, which would output $((r_x,r_z), X^{r_x}Z^{r_z}P^xTX^xZ^z|\alpha\rangle)$. The classical bits $(r_x,r_z)$ are the partial inputs for the key-updating procedure, and the quantum state is the encrypted result. There are three transmissions of messages in the whole scheme. The first and third messages are encrypted data and encrypted result, respectively. So there is only one transmission in the homomorphic evaluation. Concretely, there is an unidirectional transmission from Server to Client in the execution of $EG[P^x]$.

The two bits $x,z\in\{0,1\}$ are Client's secret key, and are unknown by Server. Client's quantum measurement depends on the secret bit $x$ and the measurement basis is $\{|\Phi(P^x)_{ab}\rangle\}_{a,b\in\{0,1\}}$. In addition, Client's outputs $r_x,r_z$ in the measurement are completely random, and would not be sent to Server, so Server cannot obtain the secret key $x,z$ and the new key $(x',z')$, where $x'=x\oplus r_x, z'=x\oplus z\oplus r_z$.

\subsection{Analysis}
In the above scheme, the implementation of $EG[P^x]$ (in the dashed circle) contains quantum measurement. According to the principal of deferred measurement, Client could defer the $P^x$-rotated Bell measurement until Server returns the result of homomorphic evaluation. While measuring the qubit $s$ and qubit $c$, Client obtains the encrypted result $X^{x'}Z^{z'}T|\alpha\rangle$ and two random bits $(r_x,r_z)$. Then Client can update the secret key $(x,z)$ according to $(r_x,r_z)$ and obtain a new secret key $(x',z')$, which is used to decrypt the encrypted result.

In Fig.\ref{fig3}, there are three arrows, which represent the 1st, 2nd, and 3rd information transmission, respectively. Both the 2nd and 3rd transmissions are from Server to Client. Because the quantum measurement can be deferred until Server returns the final result of homomorphic evaluation, the 2nd message can be incorporated into the 3rd transmission. Then the whole scheme contains only two information transmissions, and no interaction is necessary during the phase of homomorphic evaluation.

\begin{proposition}
The QHE scheme for $T$-gate is correct. That is, $T$-gate can be implement homomorphically.
\end{proposition}
\begin{proof}
The phase of evaluation contains a $\mathrm{SWAP}$ operation, and the phase of decryption contains a quantum measurement. The two operations have just finished the encrypted gate $EG[P^x]$. It is easy to verify the following equation (up to a global phase):
\begin{equation}
P^xTX^xZ^z=X^xZ^{x\oplus z}T,\forall x,z\in\{0,1\}.
\end{equation}
So (up to a global phase)
\begin{eqnarray}
EG[P^x](TX^xZ^z|\alpha\rangle)&=&\left((r_x,r_z),X^{r_x}Z^{r_z}P^xTX^xZ^z|\alpha\rangle\right) \\
&=&\left((r_x,r_z),X^{r_x}Z^{r_z}X^xZ^{x\oplus z}T|\alpha\rangle\right) \\
&=&\left((r_x,r_z),X^{x\oplus r_x}Z^{x\oplus z\oplus r_z}T|\alpha\rangle\right),
\end{eqnarray}
where $r_x,r_z\in\{0,1\}$ and $|\alpha\rangle$ is arbitrary qubit. While finishing the $EG[P^x]$, Client would obtain $(r_x,r_z)$ and Server would obtain a qubit in the encrypted state $X^{x'}Z^{z'}T|\alpha\rangle$, where the new key $(x',z')=(x\oplus r_x, x\oplus z\oplus r_z)$; However, Server does not know the value of the new key $(x',z')$. Then Server sends the encrypted state to Client. Finally, Client firstly computes the new key $(x',z')$ from the his/her secret key $sk=(x,z)$ and measurement result $r_x,r_z$ and then decrypts the encrypted state and obtains plaintext result $T|\alpha\rangle$.~~\hfill{}$\Box$
\end{proof}

\begin{proposition}
The QHE scheme for $T$-gate is perfectly secure.
\end{proposition}
\begin{proof}
The security is guaranteed by the security of QOTP and $EG[P^x]$. Concretely, Client performs QOTP encryption on quantum data $|\alpha\rangle$ and sends the encrypted data to Server, so Server cannot obtain any information about the plaintext data $|\alpha\rangle$ and random bits $x,z$. Next, Server collaborates with Client and finishes the computation of $EG[P^x]$. Because Proposition \ref{prp1} has proved the perfect security of the two-party computation scheme for $EG[P^x]$, Server cannot obtain any information about $x$ and $r_x,r_z$ during the computation of $EG[P^x]$. Above all, Server cannot obtain any information about $|\alpha\rangle$, $(x,z)$ and $(r_x,r_z)$, so the final key $(x',z')$ is also perfectly protected. Then, with respect to Server, the computation result $T|\alpha\rangle$ is perfectly hidden in the encrypted result $X^{x'}Z^{z'}T|\alpha\rangle$.~~\hfill{}$\Box$
\end{proof}

It can be seen from the QHE scheme for $T$-gate that, Server only performs $T$-gate and $\mathrm{SWAP}$, and Client performs more complex quantum computation (e.g. $X$,$Z$,$CNOT$,$P$,$H$,$Z$-basis measurement); However, this QHE scheme is only customized to the homomorphic evaluation of $T$-gate. For a general quantum circuit which contains lots of different gates, we will construct its QHE scheme and compare Client's quantum complexity and Server's.

\begin{remark}
For a $T^\dagger$-gate, its QHE scheme can be constructed in the same way, and the proofs of correctness and security are also the same. The only difference lies in that, key-updating algorithm will obtain the new key $(x',z')=(x\oplus r_x, z\oplus r_z)$.
\end{remark}

In the next section, QHE schemes for any unitary quantum circuit are presented based on the scheme that implements $T/T^\dagger$-gate homomorphically. Once that Server performs a $T/T^\dagger$-gate, he/she should collaborate with Client and finishes an execution of $EG[P^x]$ (the bit $x\in\{0,1\}$ is an intermediate value during the process of key-updating. It is only known by Client, so the collaboration of Server and Client is necessary for the execution of $EG[P^x]$). Thus, once a $T/T^\dagger$-gate is performed, an unidirectional transmission from Server to Client is required. In the whole QHE scheme for quantum circuit, the total number of transmissions is the same as that of $T/T^\dagger$-gate in the circuit. In fact, all these transmissions can be deferred until finishing all the homomorphic evaluation of the circuit, then the interactions in the homomorphic evaluation procedure can be completely avoided.

\section{Quantum homomorphic encryption scheme for any unitary quantum circuit}
In the QHE scheme for $T/T^\dagger$-gate, based on $EG[P^x]$, a non-interactive method is proposed to remove $P$-error generated in the homomorphic evaluation. Moreover, using the method, we devise two QHE schemes for any unitary quantum circuit. Our schemes have the following properties: non-interaction, perfect security, homomorphic evaluation of any unitary quantum circuit, and quasi-compactness. The first scheme \texttt{GT} achieves $M\log M$-quasi-compactness, and its variation \texttt{VGT} achieves $M$-quasi-compactness.

\subsection{Main idea}\label{sec41}
According to the QHE scheme for $T$-gate, once that Server performs a $T$-gate, Server and Client should finish a two-party computation scheme of $EG[P^x]$. The implementation of $EG[P^x]$ contains quantum measurement, and the measurement depends on the secret bit $x$. Given a quantum circuit $\mathcal{C}$ which has $M$ $T/T^\dagger$-gates, if each $T/T^\dagger$-gate is implemented homomorphically following the previous scheme, then the QHE scheme for $\mathcal{C}$ would contain $M$ quantum measurements. We prove all the $M$ measurements can be deferred until Server has finished the evaluation of quantum circuit $\mathcal{C}$.

Without loss of generality, we consider a quantum circuit with at least one $T$-gate (see Fig.\ref{fig4}). The circuit contains a $T$-gate acting on the 1st qubit, and then any sub-circuit $\Omega$ acting on $n(n\geq 1)$ qubits, where $\Omega$ may contains more $T$-gates. $\Omega$ may be only one quantum gate, or a series of gates; $\Omega$ maybe act on only one qubit (if $n=1$, then $\Omega$ acts on only one qubit, and acts on the same qubit as $T$-gate), or act on $n$ qubits.

\begin{figure}[hbt]
\centering
  \includegraphics[scale=1.1]{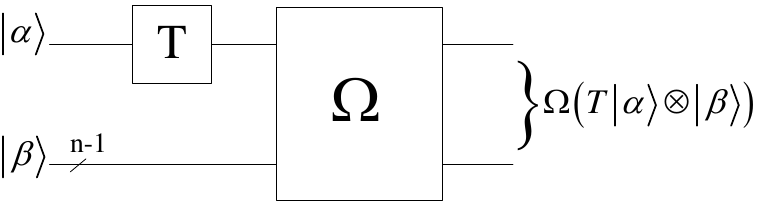}
\caption{The general structure of quantum circuit containing $T$-gates. The sub-circuit $\Omega$ contains lots of quantum gates.}
\label{fig4}       
\end{figure}

Given a quantum circuit containing $T$-gates (see Fig.\ref{fig4}), we consider that the $T$-gate is implemented homomorphically using the previous scheme. Concretely, the encrypted data $X^xZ^z|\alpha\rangle$ is transformed by $T$, and then $EG[P^x]$ performs on it and Bell state. We prove that, the measurement ($P^x$-rotated Bell measurement) in the module $EG[P^x]$ can be deferred until finishing the computation of the quantum circuit $\Omega$ (see Fig.\ref{fig5}).
\begin{figure}[hbt]
\centering
  \includegraphics[scale=1.1]{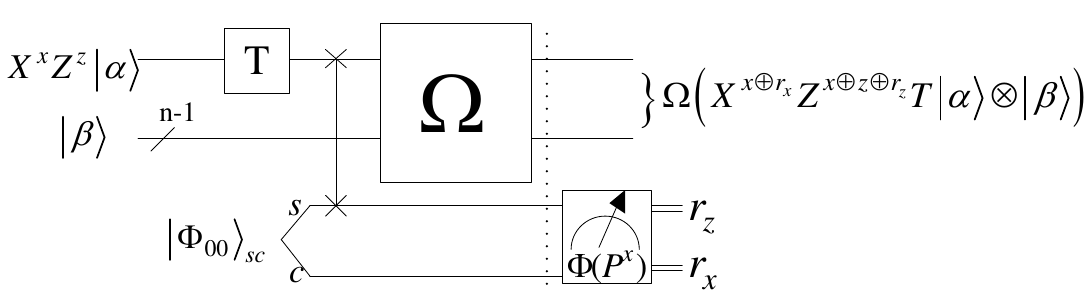}
\caption{QHE scheme for $T$-gate is used in the homomorphic evaluation of a general circuit. The $P^x$-rotated Bell measurement is deferred until the quantum circuit $\Omega$ has been performed. The measurement basis $\Phi(P^x)$ depends on the secret bit $x$, which is used to produce the encrypted data $X^xZ^z|\alpha\rangle$.}
\label{fig5}       
\end{figure}

In Fig.\ref{fig5}, there are four quantum registers. In order to explicitly describe the following deduction, these registers are denoted as $1,2,s,c$, respectively. The 2nd register has $n-1$ qubits. In the following formulas, the subscript of a quantum operator indicates which register is performed by the quantum operator, and the subscript of a quantum state indicates which register is in that state.
\begin{eqnarray}
&& (X^xZ^z|\alpha\rangle)_1\otimes |\beta\rangle_2 \otimes |\Phi_{00}\rangle_{s,c} \nonumber\\
&& \xrightarrow{T_1}(TX^xZ^z|\alpha\rangle)_1\otimes |\beta\rangle_2 \otimes |\Phi_{00}\rangle_{s,c} \\
&& \xrightarrow{SWAP_{1,s}}(TX^xZ^z|\alpha\rangle)_s\otimes |\beta\rangle_2 \otimes |\Phi_{00}\rangle_{1,c} \\
&& \xrightarrow{\Omega_{1,2}}(\Omega_{1,2}\otimes I_s \otimes I_c)((TX^xZ^z|\alpha\rangle)_s \otimes |\Phi_{00}\rangle_{1,c}\otimes |\beta\rangle_2) \\
&& =(\Omega_{1,2}\otimes I_s \otimes I_c)\sum_{a,b\in\{0,1\}}|\Phi(P^x)_{ab}\rangle_{s,c}\otimes X^aZ^bP^xTX^xZ^z|\alpha\rangle_1\otimes |\beta\rangle_2 \\
&& =\sum_{a,b\in\{0,1\}}|\Phi(P^x)_{ab}\rangle_{s,c}\otimes \Omega_{1,2}(X^{x\oplus a}Z^{x\oplus z\oplus b}T|\alpha\rangle_1\otimes |\beta\rangle_2) \\
&& \xrightarrow{U-\texttt{rotated Bell measurement}}\\
&& \Omega\left((X^{x\oplus r_x}Z^{x\oplus z\oplus r_z}T|\alpha\rangle)\otimes |\beta\rangle\right)_{1,2}, \texttt{if the measurement outputs } r_z,r_x. \nonumber
\end{eqnarray}
Thus the final result is $\left((r_x,r_z),\Omega(X^{x\oplus r_x}Z^{x\oplus z\oplus r_z}T|\alpha\rangle\otimes |\beta\rangle)\right)$. Next we consider the condition that the quantum measurement ($P^x$-rotated Bell measurement) in $EG[P^x]$ is not deferred. Concretely, $T$-gate is firstly performed and then followed by $EG[P^x]$, e.g.$EG[P^x]TX^xZ^z|\alpha\rangle=\left((r_x,r_z),X^{x\oplus r_x}Z^{x\oplus z\oplus r_z}T|\alpha\rangle\right)$, and the sub-circuit $\Omega$ is performed finally and the first two registers come into the state $\left((r_x,r_z),\Omega(X^{x\oplus r_x}Z^{x\oplus z\oplus r_z} T|\alpha\rangle\otimes|\beta\rangle)\right)$. Thus, the final result would not be affected whether or not deferring the quantum measurement in $EG[P^x]$.

Similar to the analysis about $T$-gate, we can obtain the same conclusion for the $T^\dagger$-gate.

For a quantum circuit $\mathcal{C}$ which contains $M$ $T/T^\dagger$-gates, we can recursively make use of the above method: each $T/T^\dagger$-gate is performed and followed by a encrypted gate $EG[P^x]$, where the secret bit $x$ varies every time; There are $M$ $T/T^\dagger$-gates and $M$ encrypted gates; These encrypted gates contains $M$ quantum measurements; All these measurements can be deferred until finishing the whole circuit $\mathcal{C}$. The quantum measurement in $EG[P^x]$ depends on the updating intermediate key $x$ before the execution of $T/T^\dagger$-gate, and the key-updating rule of $T/T^\dagger$-gate depends on the measurement result in the previous measurements. So all these $M$ quantum measurements must be performed alternately with the key-updating steps. (There are two alternative ways and each one results a version of QHE scheme in this article). In addition, the key-updating steps has the same order as the series of quantum gates, so all these $M$ measurements are performed in the same order as the execution of all $M$ $T/T^\dagger$-gates.

\subsection{Scheme \texttt{GT}£ºQHE scheme based on gate teleportation}
According to the above idea, we propose a non-interactive QHE scheme, in which the encryption algorithm is QOTP-encryption. In the homomorphic evaluation of quantum circuit $\mathcal{C}$, Server carries out a series of quantum gates in circuit $\mathcal{C}$ except that performing an additional $\mathrm{SWAP}$ operation after each execution of $T/T^\dagger$-gate. There is no interaction between Client and Server during the process of homomorphic evaluation.
In the process of decryption, Client alternately performs quantum measurements and key-updating, and finally performs a QOTP-decryption. There is no interaction between Client and Server during the decryption. The detailed description is presented as follows.

Suppose there is a unitary quantum circuit $\mathcal{C}$ acting on $n$ qubits. Assume it consists of the quantum gates from the set $\mathcal{S}=\{X,Z,H,P,CNOT,T,T^\dagger\}$, and its size is $N$ (There are $N$ gates in the circuit). The quantum gates in $\mathcal{C}$ are numbered from left to right, e.g. $Gate[1]$, $Gate[2]$, $\ldots$, $Gate[N]$. Denote $Gate[j]_w$ as the $j$th quantum gate acting on the $w$th qubit. For example, $Gate[j]=H$ acts on the $w$th qubit, denote $Gate[j]_w=H_w$ or $Gate[j]=H_w$. Denote $CNOT_{w,w'}$ as a two-qubit gate $CNOT$ acting on the $w$th and $w'$th qubits, where $w$th qubit is the control). Among the $N$ gates, assume the total number of $T$-gates and $T^\dagger$-gates is
$\sharp\{j:Gate[j]=T  \texttt{ or } T^\dagger\}=M$. Every $T$-gate and $T^\dagger$-gate has its own number $j_i$ ($i\leq j_i\leq N$, $1\leq i\leq M$) in the series of gates $\{Gate[j],j=1,\ldots,N\}$. Then $Gate[j_i]=T/T^\dagger$ and $j_i<j_{i+1}$. Suppose $Gate[j_i]$ ($1\leq i\leq M$) acts on $w_i$th qubit ($1\leq w_i\leq n$). That is, $Gate[j_i]=T_{w_i}/T_{w_i}^\dagger$.

Next we describe the QHE scheme $\texttt{GT}$ for the quantum circuit $\mathcal{C}$. The encryption algorithm is the QOTP encryption transformation; the plaintext data has $n$ qubits and the secret key $sk$ has $2n$ bits. Let QOTP key $sk=(x_0,z_0),x_0,z_0\in\{0,1\}^n$. The $w$th plaintext qubit is encrypted using the secret bit $(x_0(w),z_0(w))$. Once that a quantum gate is performed on an encrypted data, the QOTP key should be updated according to the key-updating rules (see \textbf{Appendix C}). After the $j$th ($1\leq j\leq N-1$) gate is performed, a new key (denoted as $(x_j,z_j),x_j,z_j\in\{0,1\}^n$) can be obtained through key-updating; It is nether the initial key nor final key, so we call it intermediate key. After the execution of the final gate and all the key-updating, a final key $(x_{final},z_{final})$ can be obtained, and can be used in QOTP decryption. It is worth to notice that, only the server knows the quantum circuit $\mathcal{C}$ (or the series of gates in $\mathcal{C}$), so Server must generate the key-updating functions according to the key-updating rules and the gates' sequence, and the key-updating functions should be sent to Client together with the encrypted results. Client computes the final key $(x_{final},z_{final})$ using the key-updating functions.

QHE scheme $\texttt{GT}$ for quantum circuit $\mathcal{C}$ can be described as the five parts: \textbf{Setup}, \textbf{Key Generation}, \textbf{Encryption}, \textbf{Evaluation}, and \textbf{Decryption}. In our scheme, \textbf{Setup} is an additional part compared with the usual definition of QHE. It preshares Bell states between Client and Server. Actually, these Bell states can be generated during the evaluation, e.g. before each $\mathrm{SWAP}$ operation.
In the QHE scheme, the first half of the $i$th Bell state is relabeled as $c_i$, and the other as $s_i$. Then the $i$th Bell state is denoted as $|\Phi_{00}\rangle_{c_i,s_i}$, $i=1,\ldots,M$.

\noindent\fbox{
\begin{minipage}{33.3em}
\begin{center}
    \textbf{Scheme \texttt{GT}£ºQHE scheme based on gate teleportation}
\end{center}

\begin{description}
  \item[1.Setup:] $\texttt{GT.Setup}(1^M)$. Generate $M$ Bell states $\{|\Phi_{00}\rangle_{c_i,s_i},i=1,\ldots,M\}$ as the ancillary states, where qubits $c_i,i=1,\ldots,M$ and qubits $s_i,i=1,\ldots,M$ are kept by Client and Server, respectively.
  \item[2.Key Generation:] $\texttt{GT.KeyGen}(1^n)$. Generate random bits $x_0,z_0\in\{0,1\}^n$, and output $sk=(x_0,z_0)$ as the secret key.
  \item[3.Encryption:] $\texttt{GT.Enc}(sk,\sigma)$. For any $n$-qubit data $\sigma$, Client performs QOTP encryption with the secret key $sk=(x_0,z_0)$, e.g. $\sigma\rightarrow \rho=X^{x_0}Z^{z_0}\sigma Z^{z_0}X^{x_0}$.
  \item[4.Evaluation:] $\texttt{GT.Eval}\left(\mathcal{C},\{\mathrm{qubit} s_i\}_{i=1}^M,\rho\right)$. Server should carry out the following two steps.

  \begin{minipage}{31.5em}
  \begin{description}
    \item[4-1.Quantum computation.] With the help of the ancillary qubits $s_i$, $i=1,\ldots,M$, Server carries out the quantum gates $Gate[1],Gate[2],\ldots,Gate[N]$ on the $n$-qubit encrypted data $\rho$. For each $j\in\{1,\ldots,N\}$, there are two cases as follows.

        \begin{minipage}{29.7em}
        \begin{description}
            \item[4-1-a] When $j\notin\{j_1,\ldots,j_M\}$, that means $Gate[j]\notin\{T,T^\dagger\}$, then Server carries out quantum gate $Gate[j]$;
            \item[4-1-b] When $j=j_i(1\leq i\leq M)$, that means $Gate[j]=Gate[j_i]=T_{w_i}$ or $T_{w_i}^\dagger$, then Server firstly performs quantum gate $Gate[j]$ on the qubit $w_i$ and then performs swap operation $\mathrm{SWAP}(\texttt{qubit} w_i, \texttt{qubit} s_i)$.
        \end{description}
        \end{minipage}

    \item[4-2.Generating $M+1$ key-updating functions $\{g_i\}_{i=1}^M$ and $f$.]

        \begin{minipage}{29.7em}
        \begin{description}
            \item[4-2-a] According to key-updating rules, Server generates the polynomial $\{g_i\}_{i=1}^M$ for one key bit $x_{j_i-1}(w_i)\in\{0,1\}$. Denote it as
            \begin{equation}
            x_{j_i-1}(w_i)=\begin{cases}
            g_1(x_0,z_0), \texttt{if } i=1;\\
            g_i(x_0,z_0,r_x(1),r_z(1),\ldots,r_x(i-1),r_z(i-1)), \\
            \texttt{ if } i=2,\ldots,M.
            \end{cases}
            \end{equation}

            \item[4-2-b] According to key-updating rules, Server generates the polynomial $f$ for the final key $(x_{final},z_{final})\in\{0,1\}^{2n}$. Denote it as \begin{equation}(x_{final},z_{final})=f(x_0,z_0,r_x(1),r_z(1),\ldots,r_x(M),r_z(M)).\end{equation}
        \end{description}
        \end{minipage}
  \end{description}
  \end{minipage}

  After finishing the steps \textbf{(4-1)(4-2)}, Server outputs the $n$-qubit encrypted result $\rho'$ and all the ancillary qubits $s_i,i=1,\ldots,M$, together with all the key-updating functions $f$ and $\{g_i\}_{i=1}^M$.
\end{description}
\end{minipage}
}

%
%
%

\noindent\fbox{
\begin{minipage}{33.3em}
\begin{center}
    \textbf{Scheme \texttt{GT} continued}
\end{center}

\begin{description}
  \item[5.Decryption:] $\texttt{GT.Dec}\left(sk,\{g_i\}_{i=1}^M,f,\{\mathrm{qubit} s_i, \mathrm{qubit} c_i\}_{i=1}^M,\rho'\right)$. According to secret key $sk=(x_0,z_0)$ and all the key-updating functions, Client performs quantum measurements on the $M$ pairs of qubits $\{(s_i,c_i)\}_{i=1}^M$, and QOTP-decryption on the $n$-qubit encrypted result $\rho'$. The detailed procedure has the following three steps.

  \begin{minipage}{31.5em}
  \begin{description}
    \item[5-1.Alternate execution of $\{g_i\}_{i=1}^M$-computing and  measurement.] For each $i=1,\ldots,M$, Client finishes $i$th round of computation which contains the two steps.

        \begin{minipage}{29.7em}
        \begin{description}
            \item[5-1-a.Computing measurement basis $\Phi(P^b)$.] According to secret key $sk=(x_0,z_0)$ and part of measurement results $r_x(1),r_z(1),\ldots,r_x(i-1),r_z(i-1)$, Client computes a key bit $b$ from the key-updating function $g_i$
            \begin{equation}
            b=\begin{cases}
            g_1(x_0,z_0), \texttt{if } i=1;\\
            g_i(x_0,z_0,r_x(1),r_z(1),\ldots,r_x(i-1),r_z(i-1)), \\
            \texttt{ if } i=2,\ldots,M.
            \end{cases}
            \end{equation}

            \item[5-1-b.Quantum measurement.] Based on the measurement basis $\Phi(P^b)$, Client performs quantum measurement on the qubit $c_i$ and qubit $s_i$, and obtains the measurement result $m_x,m_z\in\{0,1\}$. Let $(r_x(i),r_z(i))=(m_x,m_z)$.
        \end{description}
        \end{minipage}

        Alternately perform $M$ rounds of the two steps (5-1-a) and (5-1-b), and finally obtain all the measurement results $\{r_x(i),r_z(i)\}_{i=1}^M$.

    \item[5-2.Computing the final key $(x_{final},z_{final})$.] The secret key $sk=(x_0,z_0)$ and all the measurement results $\{r_x(i),r_z(i)\}_{i=1}^M$ are inputted into the function $f$, and the computation generates the final key $(x_{final},z_{final})$.
    \item[5-3.QOTP decryption.] According to the final key $(x_{final},z_{final})$, Client performs QOTP decryption transformation on the encrypted result, e.g. $\rho'\rightarrow \sigma'=Z^{z_{final}}X^{x_{final}}\rho'X^{x_{final}}Z^{z_{final}}$.
  \end{description}
  \end{minipage}
\end{description}
\end{minipage}
}

In the scheme \texttt{GT}, the measurement basis depends on one bit of some intermediate key, and the key bit should be computed from the key-updating functions $\{g_i\}_{i=1}^M$, which relate to the measurement results. Thus, Client must alternately perform quantum measurements and key-updating functions, and obtains the final key. The number of alternations is the same as the number of $T/T^\dagger$-gates in the circuit $\mathcal{C}$.

According to the above QHE scheme, all the quantum measurements are deferred until finishing the execution of all the quantum gates $\{Gate[j]\}_{j=1}^N$. Thus, both Server and Client need $M$-qubit quantum memory in order to temporarily store qubits $s_i$ and qubits $c_i$. When Server returns the encrypted result, the qubits $s_i,i=1,\ldots,M$ are also sent back to Client. Client performs quantum measurements on each pair of qubits $(s_i,c_i)$, and the measurement basis is computed from his/her secret key $sk=(x_0,z_0)$ and the key-updating functions, which are provided by Server.

\subsection{Analysis}
The scheme \texttt{GT} is analyzed from the correctness, security and quasi-compactness.

\begin{theorem}
Scheme \texttt{GT} is correct. That is, let $\mathcal{F}$ be the class of all unitary quantum circuits consisted of quantum gates in $\mathcal{S}=\{X,Z,H,P,CNOT,T,T^\dagger\}$, the scheme \texttt{GT} is $\mathcal{F}$-homomorphic.
\end{theorem}
\begin{proof}
It is sufficient to prove that, for any $n$-qubit unitary circuit $\mathcal{C}$ and any $n$-qubit data $\sigma$,
\begin{eqnarray}
&& \Delta\left(\texttt{GT.Dec}\left(sk,\{\mathrm{qubit} c_i\}_{i=1}^M, \texttt{GT.Eval}\left(\mathcal{C},\{\mathrm{qubit} s_i\}_{i=1}^M, X^{x_0}Z^{z_0}\sigma Z^{z_0}X^{x_0} \right) \right),\Phi_\mathcal{C}(\sigma) \right) \nonumber\\
&& \leq negl(n),
\end{eqnarray}
where the quantum circuit $\mathcal{C}$ consists of the gates $Gate[1]$,$Gate[2]$,$\ldots$,$Gate[N]$, and $\sharp\{j|Gate[j]\in\{T,T^\dagger\}\}=M$.

In the scheme \texttt{GT}, for each $Gate[j_i]\in\{T,T^\dagger\}$, its homomorphic evaluation is realized by using encrypted gate $EG(P^{x_{j_i-1}(w_i)})$, where $w_i$ indicates the qubit which the $Gate[j_i]$ acts on. The implementation of $EG(P^{x_{j_i-1}(w_i)})$ contains two parts --- $\mathrm{SWAP}$ and quantum measurement, which are finished in \texttt{GT.Eval} and \texttt{GT.Dec} respectively. The measurement in \texttt{GT.Dec} depends on one bit of the intermediate key. If the intermediate key is correct, then $EG(P^{x_{j_i-1} (w_i)})$ can be implemented correctly.

In the following, we prove the intermediate keys and final key can be computed correctly.
According to Section \ref{sec41}, there exists no effect on the final result whether or not defer the measurement part of $EG(P^{x_{j_i-1}(w_i)})$. Thus, for convenient to prove our result, we assume that all the two parts of $EG(P^{x_{j_i-1}(w_i)})$ are finished in \texttt{GT.Eval} and all the measurement results $r_x(1)$,$r_z(1)$,$\ldots$,$r_x(M)$,$r_z(M)$ are obtained in \texttt{GT.Eval}.

The algorithm \texttt{GT.Eval} performs quantum computing on encrypted data $X^{x_0}Z^{z_0}\sigma Z^{z_0}X^{x_0}$ according to the order of the quantum gates $Gate[j]$,$j=1,\ldots,N$. Synchronically, \texttt{GT.Eval} derives  key-updating functions $\{g_i\}_{i=1}^M$ and $f$ according to key-updating rules. The correctness of the key-updating rules can ensure that each intermediate key $(x_i,z_i)$ and final key $(x_{final},z_{final})$ can be computed correctly from the algorithm \texttt{GT.Dec} and the initial key $sk=(x_0,z_0)$. Then the encrypted result generated by $\texttt{GT.Eval}\left(\mathcal{C},\{\mathrm{qubit} s_i\}_{i=1}^M,X^{x_0}Z^{z_0}\sigma Z^{z_0}X^{x_0}\right)$ is just $$X^{x_{final}}Z^{z_{final}}\Phi_\mathcal{C}(\sigma)Z^{z_{final}}X^{x_{final}},$$ where $(x_{final},z_{final})=f(x_0,z_0,r_x(1),r_z(1),\ldots,r_x(M),r_z(M))$. Finally the decryption would output the plaintext result $\Phi_\mathcal{C}(\sigma)$, so
\begin{eqnarray*}
&& \Delta\left(\texttt{GT.Dec}\left(sk,\{\mathrm{qubit} c_i\}_{i=1}^M, \texttt{GT.Eval}\left(\mathcal{C},\{\mathrm{qubit} s_i\}_{i=1}^M, X^{x_0}Z^{z_0}\sigma Z^{z_0}X^{x_0} \right) \right),\Phi_\mathcal{C}(\sigma) \right) \\
&& \equiv 0.
\end{eqnarray*}
~~\hfill{}$\Box$
\end{proof}

\begin{theorem}
Scheme \texttt{GT} is perfectly secure.
\end{theorem}
\begin{proof}
Firstly, the algorithm \texttt{GT.Enc} is a QOTP encryption transformation, so it can perfectly encrypt the plaintext and the secret key $sk=(x_0,z_0)$ is also hidden perfectly. During the evaluation procedure, there is no interaction between Server and Client, so Server cannot obtain any information about the plaintext and the key.

Secondly, during the decryption procedure, Client carries out some classical computation and quantum measurements, and there is no interaction in the computation. Then Client's classical computation would reveal no information. Each measurement is part of an encrypted gate, which has been proved to be perfectly secure in Proposition \ref{prp1}, thus Client's quantum measurements would reveal no information too.

Finally, Client obtains the perfectly secure final key and performs QOTP decryption which is also perfectly secure. Thus, the scheme \texttt{GT} is perfectly secure.
~~\hfill{}$\Box$
\end{proof}

\begin{theorem}
Scheme \texttt{GT} is a QHE scheme with $M\log M$-quasi-compactness.
\end{theorem}
\begin{proof}
In the algorithm \texttt{GT.Eval}, Server should generate $M+1$ key-updating functions
\begin{equation}
\begin{cases}
g_i:\{0,1\}^{2n+2(i-1)}\rightarrow\{0,1\},i=1,\ldots,M;\\
f:\{0,1\}^{2n+2M}\rightarrow\{0,1\}^{2n}.
\end{cases}
\end{equation}
Each key-updating function can be expressed with the XOR of some binary variables.
Then, the computational complexity of $g_i$ is at most $\log_2(2n+2(i-1))$, $i=1,\ldots,M$ and the computational complexity of $f$ is at most $2n\log_2(2n+2M)$. Thus, the classical complexity of decryption procedure is at most $2n\log_2(2n+2M)+\sum_{i=1}^M \log_2(2n+2(i-1))=O((M+n)\log_2(M+n))$. The decryption procedure contains $M$ quantum measurements and $n$-qubit QOTP decryption, so its quantum complexity is $M+2n$. So time complexity of \texttt{GT.Dec} is $O((M+n)\log_2(M+n))$. Thus the dependence of the complexity of \texttt{GT.Dec} on the evaluated circuit $\mathcal{C}$ is $M\log M$.
~~\hfill{}$\Box$
\end{proof}

Scheme \texttt{GT} has the following features.
\begin{enumerate}
  \item If the quantum data has $n$ qubits, then the secret key has $2n$ qubits and the encrypted data has $n$ qubits, and the encrypted result obtained by \texttt{GT.Eval} preserves the same size as the encrypted data.
  \item The number of Bell states is the same as that of $T/T^\dagger$ in $\mathcal{C}$.
  \item There are $M$ ``rotated Bell measurements'', which generate $2M$ classical bits.
  \item The evaluation key is not necessary in the algorithm \texttt{GT.Eval}.
  \item During the key-updating procedure (step \texttt{5-1}), all the intermediate keys are not stored. Every intermediate key and the final key have the same size as the secret key.
  \item Client must know the total number $M$ of $T/T^\dagger$-gates in Server's circuit $\mathcal{C}$. The number determines how many times Client would perform quantum measurements.
  \item It is unnecessary for Client to know any information of Server's quantum circuit $\mathcal{C}$, except the total number of $T/T^\dagger$-gates.
\end{enumerate}

\begin{remark}
In \texttt{GT.Setup}, Client and Server preshares $M$ Bell states, where $M$ is the total number of $T/T^\dagger$-gates. Actually, it is unnecessary for them to preshare Bell states. Instead, the Bell state can be produced by Server whenever he/she performs a $T/T^\dagger$-gate during the evaluation procedure. Moreover, Server could produce as many Bell states as required. Finally Server's qubits are all sent back to Client together with the encrypted result.
\end{remark}

\subsection{Scheme \texttt{VGT}: a variation of \texttt{GT} with optimal quasi-compactness}
In this section, the scheme \texttt{GT} is slightly modified and a variation named \texttt{VGT} is presented. Before we describe the scheme \texttt{VGT}, we prove an optimal quasi-compactness bound. Finally the scheme \texttt{VGT} can achieve this optimal bound.
\begin{theorem}
For any non-interactive and $\mathcal{F}$-homomorphic QHE scheme with perfect security, it must be quasi-compact and $M$-quasi-compactness is the optimal bound.
\end{theorem}
\begin{proof}
For any QHE scheme \texttt{QHE=\{QHE.KeyGen,QHE.Enc,QHE.Eval,QHE.Dec\}} with no interaction, Client and Server must finish the following three stages. There are only two transmissions, where the first one is in the stage 1 and the second one is in the stage 2.
\begin{description}
  \item[Stage 1.] Client performs the algorithm \texttt{QHE.Enc} on quantum data, and sends the encrypted data to Server;
  \item[Stage 2.] Server performs the homomorphic evaluation \texttt{QHE.Eval}, and sends the encrypted result to Client;
  \item[Stage 3.] Client performs the algorithm \texttt{QHE.Dec} on the encrypted result, and obtains the plaintext result.
\end{description}

All these stages are analyzed one by one.
\begin{description}
  \item[Stage 1.] Because \texttt{QHE} is required to be perfectly secure, \texttt{QHE.Enc} must adopt QOTP encryption transformation, and the secret key $sk$ would not be revealed. Thus, Server cannot obtain any information about the secret key from the first transmission.
  \item[Stage 2.] Because \texttt{QHE} is $\mathcal{F}$-homomorphic, it must allow the homomorphic evaluation of $T$-gate. In the homomorphic evaluation \texttt{QHE.Eval}, once that a $T$-gate is executed on the encrypted data, a key-dependent $P$-error would occur. If Server's quantum circuit $\mathcal{C}$ contains $M$ $T$-gates, then the algorithm \texttt{QHE.Eval} would cause $M$ $P$-errors, and each $P$-error is controlled by a different bit of certain intermediate key. In Stage 1, Server cannot obtain any information about the secret key, so he/she cannot correct the $M$ $P$-errors. Thus, the encrypted result in Stage 2 has $M$ $P$-errors, and the $M$ $P$-errors are inevitably continued to Stage 3.
  \item[Stage 3.] In the algorithm \texttt{QHE.Dec}, the $M$ $P$-errors must be corrected. Each $P$-error depends on the value of a different bit in certain intermediate key. That means, the $M$ $P$-errors are controlled by $M$ intermediate keys (Denote them as $key_j$,$j=1,\ldots,M$). According to key-updating rules, there exists logically computational relations between these different keys (Concretely there exists a family of functions $\{h_j\}_{j=1}^M$ which satisfy the relations $key_j=h_j(key_{j-1}),j\in\{1,\ldots,M-1\}$, where $key_0$ is the secret key $sk$), and these relations $\{h_j\}_{j=1}^M$ are determined by the quantum circuit $\mathcal{C}$. Thus, in the second transmission, these relations $\{h_j\}_{j=1}^M$ must be included in the message beside the encrypted result. So Client must finish at least $M$ steps of computation in order to obtains these $M$ keys, and finally carries out the correction of $M$ $P$-errors. Moreover, these keys are sequentially dependent, e.g. $key_j=h_j(key_{j-1})$, so the $M$ steps of computation cannot be executed in parallel. Instead, they must be finished one by one. Thus the time complexity of \texttt{QHE.Dec} must be dependent on $M$, and is at least $O(M)$.
\end{description}
Thus, for any non-interactive and $\mathcal{F}$-homomorphic QHE scheme with perfect security, $M$-quasi-compactness is the optimal bound.
~~\hfill{}$\Box$
\end{proof}

Next we present the variation \texttt{VGT} of the scheme \texttt{GT}.  Scheme \texttt{VGT} is a QHE scheme with $M$-quasi-compactness, and its security is the same as \texttt{GT}.

Compared to \texttt{GT}, the scheme \texttt{VGT}'s differences only exist in the evaluation and decryption procedures. So we only describe the modified parts in the algorithms \texttt{VGT.Eval} and \texttt{VGT.Dec}. Let $j_0=0$ in the following description.

The evaluation procedure of \texttt{VGT} is expressed with $\texttt{VGT.Eval}(\mathcal{C},\{\mathrm{qubit} s_i\}_\{i=1\}^M,\rho)$. Concretely, the step $(\textbf{4-2})$ in \texttt{GT.Eval} is replaced with the step $(\textbf{4-2}')$ as follows.

\noindent\fbox{
\begin{minipage}{33.3em}
\begin{description}
  \item[$\textbf{4-2}'$.] Generating $2M+1$ key-updating functions $\{g_i\}_{i=1}^M$ and $\{f_i\}_{i=1}^{M+1}$.

          \begin{minipage}{30.7em}
          \begin{description}
            \item[4-2$'$-a.] According to key-updating rules, Server generates the polynomial $\{g_i\}_{i=1}^M$ for one key bit $x_{j_i-1}(w_i)\in\{0,1\}$. Denote it as
            \begin{equation}
            x_{j_i-1}(w_i)=g_i(x_{j_{i-1}},z_{j_{i-1}}),i=1,\ldots,M.
            \end{equation}

            \item[4-2$'$-b.] According to key-updating rules, Server generates the polynomial $\{f_i\}_{i=1}^M$ for the intermediate key $(x_{j_i},z_{j_i})\in\{0,1\}^{2n}$. Denote it as \begin{equation}(x_{j_i},z_{j_i})=f_i(x_{j_{i-1}},z_{j_{i-1}},r_x(i),r_z(i)),i=1,\ldots,M.\end{equation}

            \item[4-2$'$-c.] According to key-updating rules, Server generates the polynomial $f_{M+1}$ for the final key $(x_{final},z_{final})\in\{0,1\}^{2n}$. Denote it as \begin{equation}(x_{final},z_{final})=f_{M+1}(x_{j_M},z_{j_M}).\end{equation}
        \end{description}
        \end{minipage}

  \item[]
\end{description}
\end{minipage}
}

The decryption procedure of \texttt{VGT} is expressed with $$\texttt{VGT.Dec}(sk,\{g_i\}_{i=1}^M,\{f_i\}_{i=1}^{M+1},\{\mathrm{qubit} s_i, \mathrm{qubit} c_i\}_{i=1}^M,\rho').$$ Concretely, the steps (\texttt{5-1}) and (\texttt{5-2}) in \texttt{GT.Dec} are replaced with the steps (\texttt{5-1}$'$) and (\texttt{5-2}$'$) as follows.

\noindent\fbox{
\begin{minipage}{33.3em}
\begin{description}
  \item[5-1$'$.Alternate execution of key-updating and measurement.] For $i=1$ to $M$, Client finishes the computation (\texttt{5-1$'$-a})(\texttt{5-1$'$-b})(\texttt{5-1$'$-c}) (or ``$g_i$-measurement-$f_i$" in brief). Finally Client obtains the intermediate key $(x_{j_M},z_{j_M})$.

  \begin{minipage}{30.7em}
  \begin{description}
            \item[5-1$'$-a. Compute $g_i$ and obtain measurement basis $\Phi(P^b)$.] According to the key $(x_{j_{i-1}},z_{j_{i-1}})$ and the key-updating function $g_i$, Client computes a key bit $b=g_i(x_{j_{i-1}},z_{j_{i-1}})$. (When $i=1$, $(x_{j_{i-1}},z_{j_{i-1}})$ is the secret key $sk=(x_0,z_0)$.)

            \item[5-1$'$-b. Quantum measurement.] Based on the measurement basis $\Phi(P^b)$, Client performs quantum measurement on $\mathrm{qubit} c_i$ and $\mathrm{qubit} s_i$, and obtains the measurement result $m_x$,$m_z$. Let $(r_x(i),r_z(i))=(m_x,m_z)$.

            \item[5-1$'$-c. Compute $f_i$ and update the key.] According to the key-updating function $f_i$, Client computes the intermediate key $(x_{j_i},z_{j_i})=f_i(x_{j_{i-1}},z_{j_{i-1}},r_x(i),r_z(i))$.
        \end{description}
  \end{minipage}

  \item[5-2$'$.] According to the intermediate key $(x_{j_M},z_{j_M})$, Client computes $f_{M+1}$ and obtains the final key $(x_{final},z_{final})$.
\end{description}
\end{minipage}
}

\begin{theorem}
Scheme \texttt{VGT} is a correct and perfectly secure QHE scheme.
\end{theorem}
\begin{proof}
Compare \texttt{VGT} with \texttt{GT}, the only variation is that they generate the key-updating functions in different ways. It would not affect the correctness and security of the QHE scheme, and the proof is very similar to \texttt{GT}. The details are omitted here.
~~\hfill{}$\Box$
\end{proof}

Finally, we prove the scheme \texttt{VGT} achieves the optimal quasi-compactness.

\begin{theorem}
Scheme \texttt{VGT} is a QHE scheme with $M$-quasi-compactness.
\end{theorem}
\begin{proof}
According to the scheme \texttt{VGT}, Server should generate $2M+1$ key-updating functions in the evaluation procedure.
\begin{equation}
\begin{cases}
g_i:\{0,1\}^{2n}\rightarrow\{0,1\},i=1,\ldots,M, \\
f_i:\{0,1\}^{2n+2}\rightarrow \{0,1\}^{2n},i=1,\ldots,M,  \\
f_{M+1}:\{0,1\}^{2n}\rightarrow\{0,1\}^{2n}.
\end{cases}
\end{equation}
Each key-updating function can be expressed with the XOR of some binary variables.
Then, the computational complexity of $g_i$ is at most $\log_2(2n)$, $i=1,\ldots,M$, and the complexity of $f_i$ is at most $2n\log_2(2n+2)$,$i=1,\ldots,M$. The complexity of $f_{M+1}$ is at most $2n\log_2(2n)$. Thus, the classical complexity of \texttt{VGT.Dec} is at most $M\log_2(2n)+2nM\log_2(2n+2)+2n\log_2(2n)=O(Mn\log_2 n)$. The decryption procedure contains $M$ quantum measurements and $n$-qubit QOTP decryption, so its quantum complexity is $M+2n$. Then the complexity of \texttt{VGT.Dec} is $O(Mn\log_2 n)$. Thus, the dependence of the complexity of \texttt{VGT.Dec} on the evaluated circuit $\mathcal{C}$ is $M$.
~~\hfill{}$\Box$
\end{proof}

\section{Conclusions and Discussions}
In this article, encrypted gate $EG[U]$ is defined. Based on gate teleportation, a two-party computation scheme is proposed to implement $EG[U]$ or $EG[P^x]$. Then, we present a non-interactive way to remove the $P$-error which is generated during the homomorphic evaluation of $T/T^\dagger$-gate. In this way, we construct the detailed QHE scheme \texttt{GT} and a variation \texttt{VGT}. Both of them are perfectly secure, $\mathcal{F}$-homomorphic and quasi-compact. Moreover, the schemes \texttt{GT} and \texttt{VGT} are proved to have $M\log M$-quasi-compactness and $M$-quasi-compactness, respectively. We prove the $M$-quasi-compactness is the optimal.

The QHE schemes are constructed for unitary quantum circuit that contains no quantum measurement. If we add a key-updating rule for measurement, then our schemes would be usable for the quantum circuit with measurement. The key-updating rule for measurement has been introduced in Ref.\cite{Broadbent2015}.

Our QHE schemes allow the homomorphic evaluation of any unitary quantum circuit ($\mathcal{F}$-homomorphic), and have perfect security. However, they do not conflict with the no-go result given by Yu et al.\cite{Yu2014}, since they are not compact and do not satisfy the definition of QFHE. It should be emphasized that, it is valuable to study the quasi-compact QHE scheme which can implement any quantum circuit homomorphically. Using the $M$-quasi-compact QHE scheme \texttt{VGT}, the complexity of the decryption procedure is independent of the size of quantum circuit $\mathcal{C}$; The decryption procedure is efficient only if $\mathcal{C}$ contains polynomial number of $T/T^\dagger$-gates (e.g. $M=poly(n)$). Thus, our QHE schemes are suitable for homomorphic evaluation of any quantum circuit with low $T/T^\dagger$-gate complexity, such as any polynomial-size quantum circuit or any quantum circuit with polynomial number of $T/T^\dagger$-gates.

Based on these results, there are two possible directions in the future researches.

\begin{itemize}
  \item Transform interactive quantum protocol into non-interactive one. This article has proposed a gate-teleportation-based two-party computation scheme, which can implement $EG[U]$. Actually, this scheme can be used to remotely perform $U$ gate. Because the quantum measurement in the scheme can be deferred until the final stage, the interaction can be eliminated, and then an interactive protocol can be transformed into a non-interactive one. This kind of transformation may be extended to transform many other interactive quantum protocols.
  \item Straightforward application of our QHE schemes. The efficiency of our QHE schemes depends on the total number of $T/T^\dagger$-gates. In order to implement application efficiently, the evaluated quantum circuit is required to have polynomial number of $T/T^\dagger$-gates. So it is important to optimize the number of $T/T^\dagger$-gates in the circuit while developing quantum application. In the future, we can also analyze the $T/T^\dagger$-complexity of some specific quantum algorithms (Quantum Fourier transformation (QFT) or Shor algorithm \cite{Shor1994}, HHL algorithm \cite{Harrow2009}, etc), and study the homomorphic implementation of these fundamental quantum algorithms using our QHE schemes.
\end{itemize}




\appendix
\section{Some definitions about quantum homomorphic encryption}
In this section, we introduce some concepts about QHE, including symmetric-key QHE, homomorphism, compactness, QFHE, quasi-compactness and security. Some definitions can be referred to Ref.\cite{Broadbent2015}.

\begin{definition}[Symmetric-key quantum homomorphic encryption]
A symmetric-key QHE scheme \texttt{QHE} contains the following four algorithms
\begin{center}\texttt{QHE=(QHE.KeyGen,QHE.Enc,QHE.Eval,QHE.Dec)}.\end{center}
\begin{description}
  \item[1.Key Generation.] $(sk,\rho_evk)\leftarrow \texttt{QHE.KeyGen}(1^n)$, where $sk$ is the secret key, $\rho_{evk}$ is quantum evaluation key in $D(\mathcal{H}_{evk})$. The evaluation key is optional in symmetric QHE scheme.
  \item[2.Encryption.] $\texttt{QHE.Enc}_{sk}:D(\mathcal{H}_M)\rightarrow D(\mathcal{H}_C)$, where $D(\mathcal{H}_M)$ and $D(\mathcal{H}_C)$ are the set of density operators in plaintext space and ciphertext space, respectively.
  \item[3.Evaluation.] $\texttt{QHE.Eval}^{QC}:D(\mathcal{H}_{evk}\otimes \mathcal{H}_C)\rightarrow D(\mathcal{H}_{C'}\otimes \mathcal{H}_{aux})$, where $\mathcal{H}_{C'}$ is the result space of quantum computation on the space $\mathcal{H}_C$. For any quantum circuit $QC$ (called evaluated circuit), with induced channel $\Phi_{QC}:D(\mathcal{H}_M)\rightarrow D(\mathcal{H}_{M'})$, we define a channel $\texttt{Eval}^{QC}$ that maps $D(\mathcal{H}_C)$ to $D(\mathcal{H}_{C'})$ with an additional auxiliary quantum state in space $\mathcal{H}_{aux}$. The evaluation key in $D(\mathcal{H}_{evk})$ is used up in the process.
  \item[4.Decryption.] $\texttt{QHE.Dec}_{sk}:D(\mathcal{H}_{C'}\otimes\mathcal{H}_{aux})\rightarrow D(\mathcal{H}_{M'})$. For any possible secret key $sk$, $\texttt{Dec}_{sk}$ is a quantum channel that maps ciphertext state together with auxiliary state to a plaintext state in $D(\mathcal{H}_{M'})$.
\end{description}
\end{definition}

\begin{definition}[Compactness]
QHE scheme \texttt{QHE} is compact if the algorithm \texttt{QHE.Dec} is independent of the evaluated circuit $QC$.
\end{definition}

\begin{definition}[Homomorphism]
Let $\mathcal{L}=\{\mathcal{L}_{\kappa}\}_{\kappa\in\mathds{N}}$ be a class of quantum circuits. A quantum encryption scheme \texttt{QHE} is homomorphic for the class $\mathcal{L}$ if for any sequence of circuits $\{C_{\kappa}\in \mathcal{L}_{\kappa}\}_{\kappa\in \mathds{N}}$ and input $\rho\in D(\mathcal{H}_M)$, there exists a negligible function $negl$ such that:
$$\Delta\left(\texttt{QHE.Dec}_{sk}\left(\texttt{QHE.Eval}^{C_{\kappa}}\left(\rho_{evk},\texttt{QHE.Enc}_{sk}(\rho)\right)\right),\Phi_{C_{\kappa}}(\rho)\right)\leq negl(\kappa),$$
where $(sk,\rho_{evk})\leftarrow \texttt{QHE.KeyGen}(1^{\kappa})$ and $\Phi_{C_{\kappa}}$ is the channel induced by quantum circuit $C_{\kappa}$.
\end{definition}

\begin{definition}[Quantum fully homomorphic encryption]
A QHE scheme is a quantum fully homomorphic encryption scheme if
\begin{enumerate}
  \item it is compact and
  \item it is $\mathcal{F}$-homomorphic (or homomorphic for $\mathcal{F}$), where $\mathcal{F}$ is the set of all quantum circuits over the universal quantum gate set $\{X,Z,H,P,CNOT,T,T^\dagger\}$.
\end{enumerate}
\end{definition}

\begin{definition}[Quasi-compactness]
Let $\mathcal{L}=\{L_\kappa\}_{\kappa\in\mathds{N}}$ be the set of all quantum circuits over the universal quantum gate set $\{X,Z,H,P,CNOT,T,T^\dagger\}$. Let $f:\mathcal{L}\rightarrow \mathds{R}_{\geq 0}$ be some function on the circuits in $\mathcal{L}$. A QHE scheme is $f$-quasi-compact if there exists a polynomial $p$ such that for any sequence of circuits $\{C_{\kappa}\in \mathcal{L}_\kappa\}_{\kappa\in\mathds{N}}$ with induced channels $\Phi_{C_\kappa}:D(\mathcal{H}_M)\rightarrow D(\mathcal{H}_(M'))$, the circuit complexity of decrypting the output of $\texttt{QHE.Eval}^{C_\kappa}$ is at most $f(C_\kappa)p(\kappa)$.
\end{definition}

Actually, QHE is a class of quantum encryption with special property. So its security can be defined following the definition of quantum encryption. Concretely, there are three level of security, e.g. computational security, information theoretic security and perfect security. This article focuses only on the QHE with perfect security. So we present the definition of perfect security as follows.

\begin{definition}[perfect security]
QHE scheme \texttt{QHE} is perfectly secure if there exists a quantum state $\Omega^{\mathrm{A}'}$ such that for all states $\rho^{\mathrm{AE}}$ we have that:
$$\parallel \texttt{QHE.Enc}(\rho^{\mathrm{AE}})-\Omega^{\mathrm{A}'}\otimes\rho^\mathrm{E}\parallel=0,$$
where \texttt{QHE.Enc} is an encryption algorithm performed on the part $\mathrm{A}$ of quantum state $\rho^\mathrm{AE}$. Denote $\texttt{QHE.Enc}(\rho^\mathrm{AE})$ as a quantum ensemble over the probability distribution of the secret key and all the randomness in the quantum algorithm.
\end{definition}

\section{Elementary identities about quantum gates}
For the quantum gates in the set $\mathcal{S}$, there exists the following identities (up to a global phase):
\begin{eqnarray*}
X^aZ^b &=& (-1)^{ab}Z^bX^a,\forall a,b\in\{0,1\}, \\
HX^aZ^b &=& Z^aX^bH,\forall a,b\in\{0,1\}, \\
PX^aZ^b &=& X^aZ^{a\oplus b}P,\forall a,b\in\{0,1\}, \\
P^aTX^aZ^b &=& X^aZ^{a\oplus b}T,\forall a,b\in\{0,1\}, \\
P^aT^{\dagger}X^aZ^b &=& X^aZ^bT^\dagger,\forall a,b\in\{0,1\}, \\
CNOT(X^aZ^b\otimes X^cZ^d) &=& (X^aZ^{b\oplus d}\otimes X^{a\oplus c}Z^d)CNOT,\forall a,b,c,d\in\{0,1\}.
\end{eqnarray*}

In the description of quantum circuit, the standard measurement (or $Z$-basis measurement) is represented as the diagram in Fig.\ref{fig6}(a).
The $U$-rotated Bell measurement can be transformed into the standard measurement. Their relation is shown in Fig.\ref{fig6}(b).
\begin{figure}[hbtp]
\centering
  \includegraphics[scale=1]{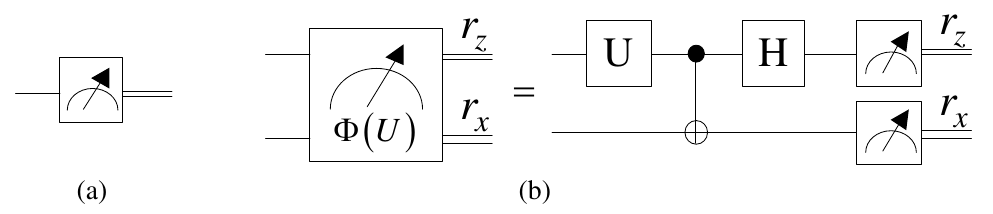}
\caption{(a)Standard measurement (or $Z$-basis measurement) in quantum circuit. (b)The relation between $U$-rotated Bell measurement and standard measurement.}
\label{fig6}       
\end{figure}

The $\mathrm{SWAP}$ can be implemented from three $\mathrm{CNOT}$ gates, e.g.
\begin{eqnarray*}
SWAP_{i,j} &=& CNOT_{i,j}CNOT_{j,i}CNOT_{i,j} \\
&=& CNOT_{i,j}(H_i\otimes H_j)CNOT_{i,j}(H_i\otimes H_j)CNOT_{i,j}.
\end{eqnarray*}
In quantum circuit, $SWAP_{i,j}$ is represented as the diagram in Fig.\ref{fig7}.
\begin{figure}[hbtp]
\centering
  \includegraphics[scale=1]{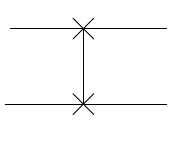}
\caption{Quantum operation $\mathrm{SWAP}$ in quantum circuit.}
\label{fig7}       
\end{figure}

\section{Key-updating rules}
Suppose the evaluated quantum circuit consists of the gates in the set $\mathcal{S}=\{X,Z,H,P,CNOT,T,T^\dagger\}$.
The quantum data has $n$ qubits and is encrypted according to QOTP scheme with the secret key being a $2n$-bit string. The secret key is the initial key of key-updating. Let the initial key be $(x_0,z_0)$, where $x_0,z_0\in\{0,1\}^n$. Denote $x_0=x_0(1)x_0(2)\cdots x_0(n)$ and $z_0=z_0(1)z_0(2)\cdots z_0(n)$. Given $n$-qubit data, the $w$th ($w=1,\ldots,n$) qubit is encrypted with the $w$th pair of bits $(x_0(w),z_0(w))$. Once that a quantum gate is performed on the $w$th encrypted qubit, the $w$th pair of bits should be updated so as to decrypt that qubit correctly. Denote by $(x_j,z_j)$ the key after the $j$th key-updating, where $x_j=x_j(1)x_j(2)\cdots x_j(n)$,$z_j=z_j(1)z_j(2)\cdots z_j(n)$.

\begin{description}
  \item[Rules 1:] If $Gate[j]$ does not act on the $w$th qubit, then let $(x_j(w),z_j(w)):=(x_{j-1}(w),z_{j-1}(w))$; Otherwise goto rules 2 and 3.
  \item[Rules 2:] If $Gate[j]=CNOT_{w,w'}$, then let
  \begin{eqnarray}
  (x_j(w),z_j(w))&&:=(x_{j-1}(w),z_{j-1}(w)\oplus z_{j-1}(w')),\\
  (x_j(w'),z_j(w'))&&:=(x_{j-1}(w)\oplus x_{j-1}(w'),z_{j-1}(w')).
  \end{eqnarray}
  \item[Rules 3:] If $Gate[j]$ acts only on the $w$th qubit, there exist the following cases.
  \begin{enumerate}
    \item If $Gate[j]=X_w$ or $Z_w$, then let $(x_j(w),z_j(w)):=(x_{j-1}(w),z_{j-1}(w))$.
    \item If $Gate[j]=H_w$, then let $(x_j(w),z_j(w)):=(z_{j-1}(w),x_{j-1}(w))$.
    \item If $Gate[j]=P_w$, then let $(x_j(w),z_j(w)):=(x_{j-1}(w),x_{j-1}(w)\oplus z_{j-1}(w))$.
    \item If $Gate[j]\in\{T_w,T_w^\dagger\}$, in our QHE schemes, $Gate[j]$ is executed with a subsequent operation $EG[P^{x_{j-1}(w)}]$. Denote by ($r_x(i)$, $r_z(i)$) the classical output of $EG[P^{x_{j-1}(w)}]$ (Assume $Gate[j]$ is the $i$th gate in the sequence of $T/T^\dagger$-gates). Let
        \begin{eqnarray*}
        (x_j(w),z_j(w)):=&& (x_{j-1}(w)\oplus r_x(i),x_{j-1}(w)\oplus z_{j-1}(w)\oplus r_z(i)), \\
        && \text{if } Gate[j]=T_w;\\
        (x_j(w),z_j(w)):=&& (x_{j-1}(w)\oplus r_x(i),z_{j-1}(w)\oplus r_z(i)), \text{if } Gate[j]=T_w^\dagger.
        \end{eqnarray*}
  \end{enumerate}
\end{description}

Based on these key-updating rules \textbf{Rules 1,2,3}, the keys $(x_j,z_j)$, $j=1,\ldots,N$ can be computed from the evaluated circuit and initial key $(x_0,z_0)$. It can be verified that, these keys satisfy the following relations.
\begin{enumerate}
  \item If $Gate[j]$ does not act on $w$th qubit ($\forall j,w$), then
  $$(x_{j-1}(w),z_{j-1}(w))=(x_j(w),z_j(w)).$$
  \item If $Gate[j]\in\{X,Z,H,P\}$ acts on $w$th qubit, then (up to a global phase)
  $$Gate[j]_w X^{x_{j-1}(w)} Z^{z_{j-1}(w)}= X^{x_j(w)} Z^{z_j(w)} Gate[j]_w.$$
  \item If $Gate[j]$ is $CNOT_{w,w'}$, then (up to a global phase)
  \begin{eqnarray*}
  && CNOT_{w,w'}(X^{x_{j-1}(w)} Z^{z_{j-1}(w)} \otimes X^{x_{j-1}(w')} Z^{z_{j-1}(w')}) \\
  && =(X^{x_j(w)} Z^{z_j(w)} \otimes X^{x_j(w')} Z^{z_j(w')})CNOT_{w,w'}.
  \end{eqnarray*}
  \item If $Gate[j]$ is $T_w/T_w^\dagger$ and is the $i$th gate in the sequence of $T/T^\dagger$-gates, then (up to a global phase)
  \begin{eqnarray*}
  && EG[P^{x_{j-1}(w)}]Gate[j]_w X^{x_{j-1}(w)} Z^{z_{j-1}(w)} \\
  && = ((r_x(i),r_z(i)),X^{x_j(w)} Z^{z_j(w)} Gate[j]_w),
  \end{eqnarray*}
  where $(r_x(i),r_z(i))$ is the classical output of $EG[P^{x_{j-1}(w)}]$.
\end{enumerate}

\section{An example for the scheme \texttt{GT}}
Given any single-qubit unitary operator $U$ and any $\epsilon>0$, it is possible to approximate $U$ to within $\epsilon$ using a circuit composed of $H$ gates and $T$-gates. In order to verify the principle of our QHE scheme \texttt{GT}, we only consider the single-qubit circuit $C_1=HTHT$, which is composed of $T$ and $HTH$ (see Fig.\ref{fig8}). Up to a global phase, the gates satisfy $T=R_z(\pi/4)$ and $HTH=R_x(\pi/4)$.

\begin{figure}[hbt]
\centering
  \includegraphics[scale=0.9]{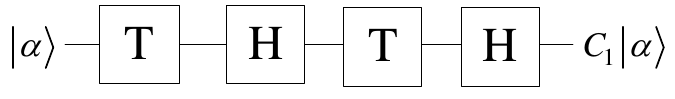}
\caption{A single-qubit quantum circuit $C_1$¡£}
\label{fig8}       
\end{figure}

The quantum circuit $C_1$ is given $n=1$ qubit as input, and the $N=4$ quantum gates are performed on the single-qubit. $C_1$ can be described as a sequence of the gates, e.g $Gate[1]=T,Gate[2]=H,Gate[3]=T,Gate[4]=H$. It contains $M=2$ $T$-gates and $j_1=1,w_1=1,j_2=3,w_2=1$. The QHE scheme \texttt{GT} for $C_1$ uses $M=2$ Bell states denoted by $|\Phi_{00}\rangle_{c_i,s_i},i=1,2$. The qubits labeled as $s_i$ and $c_i$ are held by Server and Client, respectively.
Fig.\ref{fig9} shows the QHE scheme \texttt{GT} for $C_1$. Because $n=1$, the secret key is $sk=(x_0,z_0)$, where $x_0=x_0(1),z_0=z_0(1)$.

\begin{figure}[hbt]
\centering
  \includegraphics[scale=1.25]{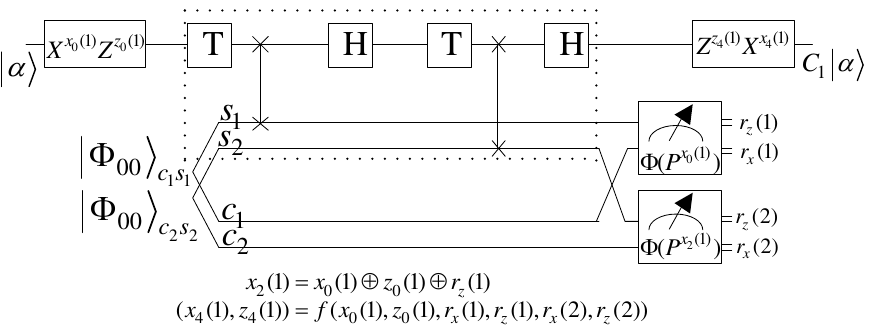}
\caption{QHE scheme \texttt{GT} for single-qubit circuit $C_1$. Server's quantum operations are shown in the dashed box, and Client's operations are shown outside the dashed box. Client's secret key is $(x_0(1),z_0(1))$, and the final key $(x_4(1),z_4(1))$ is computed from the secret key and the measurement results.}
\label{fig9}       
\end{figure}

In the evaluation procedure, Server should finish the quantum operations shown in the dashed box. In addition, Server must generate $M+1=3$ key-updating functions $g_1$, $g_2$ and $f$ based on key-updating rules and evaluated circuit $C_1$. Key-updating functions $x_0(1)=g_1(x_0,z_0)$, $x_2(1)=g_2(x_0,z_0,r_x(1),r_z(1))$ are expressed as follows
\begin{eqnarray*}
x_0(1) &=& x_0(1), \\
x_2(1) &=& x_0(1)\oplus z_0(1) \oplus r_z(1).
\end{eqnarray*}
Key-updating function $(x_4,z_4)=f(x_0,z_0,r_x(1),r_z(1),r_x(2),r_z(2))$ is expressed as follow
\begin{eqnarray*}
x_4(1) &=& z_0(1)\oplus r_x(1)\oplus r_z(1)\oplus r_z(2), \\
z_4(1) &=& x_0(1)\oplus z_0(1)\oplus r_z(1)\oplus r_x(2).
\end{eqnarray*}

In the decryption procedure, Client performs two quantum measurements: (1) according to the function $g_1$, Client computes the measurement basis $\Phi(P^{x_0(1)})$, and then performs quantum measurement and obtains a pair of bits $(r_x(1),r_z(1))$;
(2) according to the function $g_2$, Client computes the measurement basis $\Phi(P^{x_2(1)})$, and then performs quantum measurement and obtains a pair of bits $(r_x(2),r_z(2))$. Finally, according to the function $f$, Client computes the final key $(x_4(1),z_4(1))$ and performs QOTP decryption transformation.

\section{An example for the scheme \texttt{VGT}}
For two-qubit quantum computation, we choose two-qubit quantum Fourier transformation (QFT) as example. The two-qubit QFT can be implemented by the quantum circuit $C_2$ in Fig.\ref{fig10}, which contains $H,CNOT,T,T^\dagger$.

\begin{figure}[hbt]
\centering
  \includegraphics[scale=0.9]{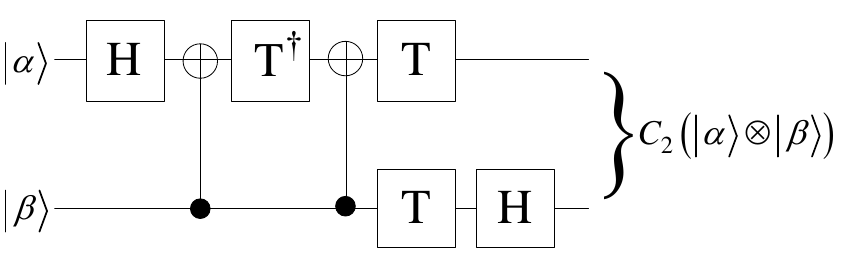}
\caption{Quantum circuit $C_2$ for the two-qubit quantum Fourier transformation.}
\label{fig10}       
\end{figure}

The quantum circuit $C_2$ is given $n=2$ qubits, and output $2$ qubits. It consists of $N=7$ quantum gates, and can be described as a sequence of gates, e.g. $Gate[1]=H_1$, $Gate[2]=CNOT_{2,1}$, $Gate[3]=T_1^\dagger$, $Gate[4]=CNOT_{2,1}$, $Gate[5]=T_1$, $Gate[6]=T_2$, $Gate[7]=H_2$. The circuit contains $2$ $T$-gates and $1$ $T^\dagger$-gate. It can be known from $C_2$ that, $M=3$ and $j_1=3$, $w_1=1$, $j_2=5$, $w_2=1$, $j_3=6$, $w_3=2$. The QHE scheme \texttt{VGT} for $C_2$ should use $M=3$ Bell states denoted by $|\Phi_{00}\rangle_{c_i,s_i}$,$i=1,2,3$. The qubits labeled as $s_i$ and $c_i$ are held by Server and Client, respectively.
Fig.\ref{fig11} shows the QHE scheme \texttt{VGT} for $C_2$.
Because $n=2$, the secret key is $sk=(x_0,z_0)$, where $x_0=x_0(1)x_0(2)$, $z_0=z_0(1)z_0(2)$.

\begin{figure}[hbt]
\centering
  \includegraphics[scale=1.25]{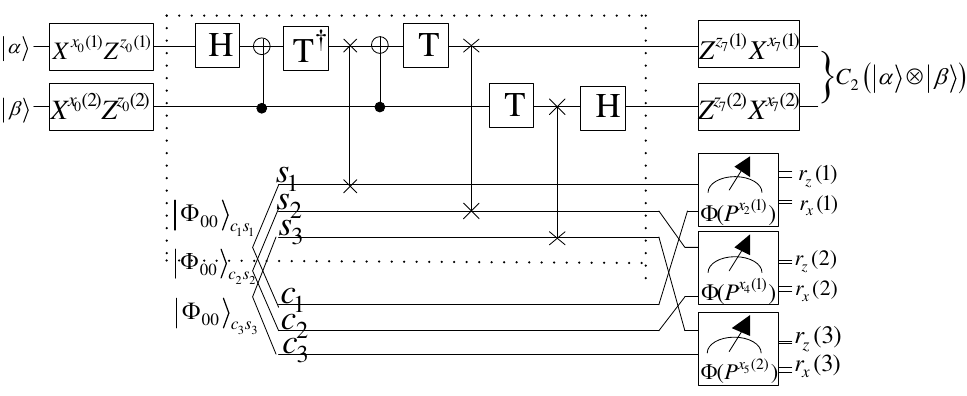}
\caption{QHE scheme \texttt{VGT} for two-qubit circuit $C_2$. Server's quantum operations are shown in the dashed box, and Client's operations are shown outside the dashed box. Client's secret key is $(x_0,z_0)$, and the final key $(x_7,z_7)$ is computed from the secret key and the measurement results.}
\label{fig11}       
\end{figure}

In the evaluation procedure, Server should finish the quantum operations shown in the dashed box. In addition, Server must generate $2M+1=7$ key-updating functions $\{g_i\}_{i=1}^3$ and $\{f_i\}_{i=1}^4$ based on key-updating rules and evaluated circuit $C_2$. Key-updating functions $x_2(1)=g_1(x_0,z_0)$, $x_4(1)=g_2(x_3,z_3)$, $x_5(2)=g_3(x_5,z_5)$ are expressed as follows
\begin{eqnarray*}
  g_1:&& x_2(1) = x_0(2) \oplus z_0(1), \\
  g_2:&& x_4(1) = x_3(1) \oplus x_3(2), \\
  g_3:&& x_5(2) = x_5(2).
\end{eqnarray*}
The three key-updating functions $(x_3,z_3)=f_1(x_0,z_0,r_x(1),r_z(1))$, $(x_5,z_5)=f_2(x_3,z_3,r_x(2),r_z(2))$, $(x_6,z_6)=f_3(x_5,z_5,r_x(3),r_z(3))$, $(x_7,z_7)=f_4(x_6,z_6)$ are expressed as follows
\begin{eqnarray*}
  f_1: && \left(\begin{array}{cc} x_3(1) & z_3(1) \\  x_3(2) & z_3(2) \end{array} \right)
        = \left(\begin{array}{cc} x_0(2)\oplus z_0(1)\oplus r_x (1) & x_0(1)\oplus r_z(1) \\ x_0(2) & x_0(1)\oplus z_0(2)  \end{array} \right) \\
  f_2: && \left(
            \begin{array}{cc}
              x_5(1) & z_5(1) \\
              x_5(2) & z_5(2) \\
            \end{array}
          \right)
        = \left(
            \begin{array}{cc}
              x_3(1)\oplus x_3(2)\oplus r_x(2) & x_3(1)\oplus x_3(2)\oplus z_3(1)\oplus r_z(2) \\
              x_3(2) & z_3(1)\oplus z_3(2) \\
            \end{array}
          \right)                       \\
  f_3: && \left(
            \begin{array}{cc}
              x_6(1) & z_6(1) \\
              x_6(2) & z_6(2) \\
            \end{array}
          \right)
   = \left(
       \begin{array}{cc}
         x_5(1) & z_5(1) \\
         x_5(2)\oplus r_x (3) & x_5(2)\oplus z_5(2)\oplus r_z(3) \\
       \end{array}
     \right)                            \\
  f_4: && \left(
            \begin{array}{cc}
              x_7(1) & z_7(1) \\
              x_7(2) & z_7(2) \\
            \end{array}
          \right)
   = \left(
       \begin{array}{cc}
         x_6(1) & z_6(1) \\
         z_6(2) & x_6(2) \\
       \end{array}
     \right)
\end{eqnarray*}

In the decryption procedure, Client performs three rounds of the computations ``$g_i$-measurement-$f_i$" and obtains the key $(x_6,z_6)$. In the $i$th round $(i=1,\ldots,3)$, Client alternately performs the following steps: (1)according to the function $g_i$, Client computes the measurement basis; (2)Client measures the pair of qubits $(s_i,c_i)$ and obtains two bits $(r_x(i),r_z(i))$; (3)according to the function $f_i$, Client computes the intermediate key $(x_{j_i},z_{j_i})$. Then, based on the function $f_4$, Client computes the final key $(x_7,z_7)$ from $(x_6,z_6)$. Finally, Client performs QOTP decryption with the key $(x_7,z_7)$.

\end{document}